\documentclass[a4paper,11pt]{amsart}

\usepackage[top=2cm, bottom=2cm, left=2cm, right=2cm]{geometry}

\usepackage[ colorlinks = true,
             linkcolor = blue,
             urlcolor  = blue,
             citecolor = red,
             anchorcolor = green,
]{hyperref}

\usepackage{qittools}
\usepackage{tabularx}

\def\idty{{\leavevmode\rm 1\mkern -5.4mu I}} 
\def\id{{\rm id}}                            
\def\Nl{{\mathbb N}}     
\def\norm #1{\Vert #1\Vert}

\def\bra #1{\langle #1\vert}
\def\ket #1{\vert #1\rangle}

\def\tr{{\rm Tr}}
\newcommand\Scp[2]{\ensuremath{\, \langle #1 \,\vert #2 \,\rangle}}

\newcommand*{\vphi}{\varphi}

\newcommand{\eqdef}{:=}

\newcommand{\sdp}{\mathrm{sdp}}

\newcommand{\channel}{\mathrm{S}}

\newcommand{\game}{\mathrm{\omega}}

\begin{document}

\title{Quantum bilinear optimization}

\author{Mario Berta}
\address{Institute for Quantum Information and Matter, Caltech, Pasadena, CA 91125, USA}
\email{berta@caltech.edu}

\author{Omar Fawzi}
\address{
Department of Computing and Mathematical Sciences, Caltech,
   Pasadena, CA 91125, USA \newline
LIP\footnote{UMR 5668 LIP - ENS Lyon - CNRS - UCBL - INRIA, Universit\'{e} de Lyon}, \'{E}cole Normale Sup\'{e}rieure de Lyon, Lyon, 69007, France}
\email{omar.fawzi@ens-lyon.fr}

\author{Volkher B.~Scholz}
\address{Institute for Theoretical Physics, ETH Zurich, 8093 Z\"urich, Switzerland}
\email{scholz@phys.ethz.ch}

\begin{abstract}
We study optimization programs given by a bilinear form over non-commutative variables subject to linear inequalities. Problems of this form include the entangled value of two-prover games, entanglement-assisted coding for classical channels and quantum-proof randomness extractors. We introduce an asymptotically converging hierarchy of efficiently computable semidefinite programming (SDP) relaxations for this quantum optimization. This allows us to give upper bounds on the quantum advantage for all of these problems. Compared to previous work of Pironio, Navascu{\'e}s and Ac{\'\i}n, our hierarchy has additional constraints. By means of examples, we illustrate the importance of these new constraints both in practice and for analytical properties. Moreover, this allows us to give a hierarchy of SDP outer approximations for the completely positive semidefinite cone introduced by Laurent and Piovesan.
\end{abstract}

\maketitle


\section{Introduction}\label{sec:intro}

\subsection{Setting}\label{sec:setting}

A major goal in quantum information theory is to understand the advantage over classical protocols that can be achieved by allowing quantum protocols. For a given information processing task, identifying the optimal success rate for this task can be seen as an optimization over the set of valid protocols. The quantum advantage is then defined as the increase in the optimal value by allowing a larger set of protocols that make use of quantum theory. A family of tasks for which such an advantage is very well-studied is the family of games between multiple parties that are not allowed to communicate. As was first demonstrated by Bell~\cite{Bel64,CHSH}, there exist games for which entanglement between the players can increase the success probability beyond the ultimate limit of classical protocols. The fundamental limit for classical protocols is called a Bell inequality and its violation indicates an important feature of quantum theory called non-locality. The topic of non-locality has been a very active topic in quantum information theory and in the foundations of quantum mechanics; see~\cite{ReviewNL} for a review. A quantum advantage can also be studied in many other settings including communication complexity~\cite{ReviewCC}, communication over a classical channel~\cite{CLMW10,PLMKR11} or randomness extractors~\cite{Renner05}. One objective of this paper is to formulate many of these problems in a unified language as bilinear optimization programs.

To make the discussion more concrete, we consider a specific example. Let $W_{X \to Y}$ be a noisy channel mapping system $X$ to system $Y$. Assuming $X$ and $Y$ are discrete systems, we can describe the channel by the transition probabilities  $W_{X \to Y}(y | x)$ from $x$ to $y$ for all $(x,y) \in X \times Y$. The goal is to send $k$ bits of information using this channel while minimizing the error probability for the decoding. A valid protocol in this setting is given by an encoding function $e : [2^k] \to X$ and a decoding function $d : Y \to [2^k]$. To take into account the possibility of a randomized functions, we describe the encoder by a probability distribution $\{e(x|i)\}_{x}$ on $X$ for every possible input $i \in [2^k]$, and similarly the decoder by a distribution $\{d(i|y)\}_{i}$ on $[2^k]$ for every $y \in Y$. Given an encoder and a decoder, the average success probability of our protocol can be expressed as
$\frac{1}{2^k} \sum_{x,y,i} d(i|y) W_{X \to Y}(y|x) e(x|i)$. In summary, optimizing the success probability of information transmission is captured by the following bilinear program
\begin{align}\label{eq:channel_classical}
\begin{aligned}
& \underset{(e,d)}{\text{maximize}}
& & \frac{1}{2^k} \sum_{x,y,i} W_{X \to Y}(y|x) d(i|y)e(x|i) \\
& \text{subject to}
& &\sum_{x} e(x|i) = 1 \quad \forall i \in [2^k] \\
& 
& &\sum_{i} d(i|y) = 1 \quad \forall y \in Y \\
& 
& & 0 \leq e(x|i) \leq 1 \quad \forall (x,i) \in X \times [2^k] \\ 
&
& & 0 \leq d(i|y) \leq 1 \quad \forall (i,y) \in [2^k] \times Y\,.
\end{aligned}
\end{align}
Observe that allowing the encoder and the decoder to access (unlimited) shared randomness does not change the optimal value of this program. A fundamental question is to study the effected of shared entanglement for communication. It is not possible to communicate only using shared entanglement between the sender and the receiver. However, shared entanglement can offer important advantages for communication if we already have a quantum channel~\cite{BBPSSW96,BW92,BSST99} or even a classical channel~\cite{CLMW10,PLMKR11}. The latter is the setting we consider here. A quantum protocol is described by a Hilbert space $\cH$ (of arbitrary dimension), a unit vector (called state) $\ket{\psi} \in \cH \otimes \cH$ shared between the encoder and the decoder, and positive operator-valued measures on $\cH$ for the encoder $\{E(x|i)\}_{x}$ for each $i \in [2^k]$ and for the decoder $\{D(i|y)\}_{i}$ for each $y \in Y$. In the quantum setting, optimizing the success probability for transmitting $k$ bits is given by
\begin{align}\label{eq:channel_quantum}
\begin{aligned}
& \underset{(\cH,\ket{\psi},E,D)}{\text{maximize}}
& & \frac{1}{2^k} \sum_{x,y,i} W_{X \to Y}(y|x) \bra{\psi} E(x|i) \otimes D(i|y) \ket{\psi} \\
& \text{subject to}
& &\sum_{x} E(x|i) = \id_{\cH} \quad \forall i \in [2^k] \\
& 
& &\sum_{i} D(i|y) = \id_{\cH} \quad \forall y \in Y \\
& 
& & 0 \preceq E(x|i) \preceq \id_{\cH} \quad \forall (x,i) \in X \times [2^k] \\ 
&
& & 0 \preceq D(i|y) \preceq \id_{\cH} \quad \forall (i,y) \in [2^k] \times Y\,.
\end{aligned}
\end{align}
Here, $\bra{\psi}$ is the conjugate transpose of the vector $\ket{\psi}$ and we write $D \preceq E$ if the operator $E-D$ is positive semidefinite.
As we can always take $\cH = \mathbb{C}$, any feasible solution for \eqref{eq:channel_classical} is also a feasible solution for \eqref{eq:channel_quantum}.

Allowing for quantum protocols also leads to the definition of quantum graph parameters~\cite{CMNSW07,RM12,LP13,BLP15}. For example, the stability number of a graph $G$ can be viewed in terms of the success probability of a two-prover game depending on $G$, or in terms of the success probability for information transmission over a noisy channel defined by $G$. Allowing quantum protocols in these tasks naturally leads to the definition of quantum stability numbers of a graph. To study such quantum graph parameters, Laurent and Piovesan~\cite{LP13} recently introduced a non-commutative analog of the completely positive cone $\mathcal{CP}$ called the completely positive semidefinite cone $\mathcal{CS}_+$. For the aforementioned problems, the set of quantum strategies can then be described using $\mathcal{CS}_+$, and the quantum advantage is witnessed by the fact $\mathcal{CS}_+$ is larger than $\mathcal{CP}$.

Having phrased the setup, let us now give a short overview of our findings.



\subsection{Results}\label{sec:objectives}

We start by phrasing problems like the ones stated above as optimization programs. More precisely, we study the class of tasks that can be described by optimizing a bilinear function subject to linear inequalities. The optimization over classical protocols corresponds to a program similar to \eqref{eq:channel_classical} with commutative (scalar) variables, whereas the optimization over quantum protocols corresponds to allowing the variables to be operator-valued as in~\eqref{eq:channel_quantum}. As it appears from the expression, optimization over quantum protocol seems quite complicated. In fact, as there is no bound on the dimension of the Hilbert space, it is not known whether the optimal value is even computable. In the context of games, Navascu{\'e}s, Pironio and Ac{\'\i}n (NPA)~\cite{Navascues07}  introduced a family of semidefinite programming (SDP) relaxations that give efficiently computable upper bounds on quantum bilinear programs. This hierarchy was shown to asymptotically converge to the optimal quantum protocol~\cite{NPA08,PNA10,Doherty:2008tm}. These hierarchies can be seen as non-commutative versions of the sum-of-squares hierarchies introduced by Lasserre and Parrilo~\cite{Las01,Par03}.

Our first contribution is the observation that many information processing tasks can be formulated in this way. We believe that phrasing these seemingly different problems in a unified language will help in our understanding of each one of these problems. Moreover, we think that tools developed in the context of optimization should be valuable in characterizing the power and limitations of quantum protocols. Our second contribution is to give a new hierarchy of SDPs that gives upper bounds on quantum bilinear programs. Compared to the previous contributions~\cite{Navascues07,NPA08,PNA10,Doherty:2008tm}, our hierarchy has some additional constraints which we illustrate to be useful in several settings. For example, the first level of our hierarchy has the nice property of being naturally bounded by the maximal value of the general problem, i.e., it is bounded by one for the case of channel coding discussed above. In addition, by means of a specific example, we show that our SDPs can give better bounds in practice. The new constraints are also important to study the completely positive semidefinite cone $\mathcal{CS}_+$, which consists of all the symmetric matrices that admit a Gram representation by positive semidefinite matrices of any size. In fact, we show that these constraints lead to a natural hierarchy of SDP outer approximations for the completely positive semidefinite cone $\mathcal{CS}_+$.




\subsection{Organization of the Paper}\label{sec:organization}

In Section~\ref{sec:quantum}, we introduce the general setup of quantum bilinear optimization and present our new hierarchy of SDPs. We keep the main text elementary and only prove that our SDPs give upper bounds on the quantum programs when the Hilbert space is finite-dimensional (the infinite-dimensional case as well as the convergence of the hierarchy are deferred to appendices). In Section~\ref{sec:examples}, we describe applications to two-prover games, channel coding, randomness extractors as well as to the optimization over the completely positive semidefinite cone. 


\section{Bilinear Optimization}\label{sec:quantum}

\subsection{Setup}

As motivated in~\eqref{eq:channel_classical} we would like to start from the following type of (classical) bilinear optimization program with real variables $z_{\alpha}$ for $\alpha\in[N]\eqdef\{1,\ldots,N\}$ and $y_{\beta}$ for $\beta\in[M]\eqdef\{1,\ldots,M\}$,\footnote{Here and henceforth we write maximize for taking the supremum (in particular the maximum might not be attained).\label{ft:maximize}}
\begin{align}\label{eq:program_classical}
\begin{aligned}
p[A,\cG,\cK]\eqdef\;&\underset{(z_\alpha,y_\beta)}{\text{maximize}}
& & \sum_{\alpha,\beta} A_{\alpha,\beta} z_\alpha y_\beta  \\
& \text{subject to}
& &g(z_1,\dots,z_N) \geq 0 \quad \forall g\in\cG\\
& & &k(y_1,\dots,y_M) \geq 0 \quad \forall k\in\cK\,.
\end{aligned}
\end{align}
with sets of affine constraints $\cG\eqdef\{g(z_1,\ldots,z_N)\}$ and $\cK\eqdef\{k(y_1,\ldots,y_M)\}$, where
\begin{align}
g(z_1,\ldots,z_N)\eqdef g^0+\sum_{\alpha\in[N]}g^\alpha z_\alpha\quad\mathrm{and}\quad k(y_1,\ldots,y_M)\eqdef k^0+\sum_{\beta\in[M]}k^\beta y_\beta\,.
\end{align}
For convenience we also define the complete set of constraints
\begin{align}\label{eq:set_constraints}
\cF\eqdef\cG\cup\cK\cup\{1\}
\end{align}
where $1$ is the function always equal to $1$. Moreover, call
\begin{align}
p[A,\cF]\eqdef p[A,\cG,\cK]
\end{align}
the classical value of~\eqref{eq:program_classical}. We restrict ourselves to affine constraints as all our applications have this form. It is however possible to extend the approach to polynomial equality constraints and have a linear term in the objective function, see Appendix~\ref{sec:generalization}.

In analogy to~\eqref{eq:channel_quantum} the corresponding quantum bilinear optimization program of~\eqref{eq:program_classical} is then as follows. Let $\cH$ be a Hilbert space (of arbitrary dimension), $\ket{\psi}\in\cH$ with $\|\ket{\psi}\|=1$, and let $E_{\alpha},D_{\beta}$ be Hermitian operators in the algebra $\cB(\cH)$ of bounded linear operators on $\cH$. By substituting the variables $z_\alpha$ in the linear constraints with operators $E_\alpha$ (and similarly for $y_\beta$ with $D_{\beta}$) we set
\begin{align}\label{eq:program_quantum}
\begin{aligned}
p^*[A,\cG,\cK]\eqdef\;&\underset{(\cH,\ket{\psi},E_\alpha,D_\beta)}{\text{maximize}}
& & \sum_{\alpha,\beta} A_{\alpha,\beta} \bra{\psi}E_{\alpha}D_{\beta}\ket{\psi}  \\
& \text{subject to}
& &[E_{\alpha},D_{\beta}]=0\quad \forall (\alpha, \beta) \in [N]\times [M]\\
& & &g(E_1,\dots,E_N) \succeq 0\quad\forall g\in\cG\\
& & &k(D_1,\dots,D_M) \succeq 0\quad \forall k\in\cK\,,
\end{aligned}
\end{align}
where $[E_{\alpha},D_{\beta}]\eqdef E_{\alpha}D_{\beta}-D_{\beta}E_{\alpha}$ denotes the commutator, and $g(E_1,\dots,E_N) \succeq 0$ means that the operator $g(E_1,\dots,E_N)$ is positive semidefinite (and similarly for $k(D_1,\dots,D_M)\succeq0$). We note that we do not think of the commutation conditions $[E_{\alpha},D_{\beta}]=0\quad \forall (\alpha, \beta) \in [N]\times [M]$ as being constraints, but rather being part of the ``quantization procedure'' itself. This is motivated by our examples originating from information theory, and the commutation relations naturally lead to their quantum versions. Moreover, from now on we assume that the sets of constraints $\cG$, $\cF$ satisfy the following.

\begin{assumption}\label{assmp:boundedop}
The set of constraints $\cG$, $\cF$ imply that there exists a positive constant $C>0$ such that the relations $-C\idty \preceq E_\alpha \preceq C \idty$ and $-C\idty \preceq D_\beta \preceq C \idty$ hold for all $(\alpha, \beta) \in [N]\times [M]$. Moreover, all operators denoted by $E_\alpha$ and $D_\beta$ are assumed to be self-adjoint.
\end{assumption}

We note that the Assumption above implies that the operator valued variables are always bounded operators, as the relations above together with the assumption of self-adjointness imply $\norm{E_\alpha},\norm{D_\beta} \leq C$.

In the following we call
\begin{align}
p^*[A,\cF]\eqdef p^*[A,\cG,\cK]
\end{align}
the quantum value of~\eqref{eq:program_quantum}, with the total set of constraints $\cF$ as in~\eqref{eq:set_constraints}. Clearly the quantum value is never smaller than the classical value,
\begin{align}
p[A,\cF]\leq p^*[A,\cF]\,.
\end{align}
Note that compared to the entanglement-assisted channel coding example~\eqref{eq:channel_quantum} we do not assume that the Hilbert space $\cH$ has tensor product form with $E_{\alpha}$ acting on the first factor and $D_{\beta}$ acting on the second factor, but only that $E_{\alpha}$ and $D_{\beta}$ commute. This takes into account the most general formulation of quantum mechanics~\cite{Haag92} (see also~\cite{Berta11_4} for a quantum information theory reference). However, for every feasible solution of~\eqref{eq:program_quantum} corresponding to a finite-dimensional Hilbert space, we can assume that the Hilbert space has a tensor product structure $\cH\otimes\cH$ with operators $E_{\alpha}\otimes\1$ and $\1\otimes D_{\beta}$ (instead of just $[E_\alpha,D_\beta]=0$ on a single space $\cH$); see e.g.,~\cite[Chapter 5]{Takesaki01} or for a self-contained quantum information theory reference~\cite{Scholz08}. Moreover, for the general infinite-dimensional case the optimal value of~\eqref{eq:channel_quantum} is certainly upper bounded by the optimal value of the corresponding program~\eqref{eq:program_quantum}.

\begin{remark}
Provided Connes' embedding conjecture has a positive answer~\cite{Connes:1976ta}, we can restrict the optimization in~\eqref{eq:program_quantum} to finite-dimensional Hilbert spaces (and thus of tensor product form). This was proved for the special case of bipartite games in~\cite{Scholz11,Fri12,Ozawa:2013hn}. For a proof sketch for the general case see Appendix~\ref{app:connes}.
\end{remark}

Our ultimate goal is to understand the gap between the classical value $p[A,\cF]$ and the quantum value $p^*[A,\cF]$ for operational examples of interest. For the problems that we study in this paper $p[A,\cF]$ is typically understood but estimating $p^*[A,\cF]$ is the challenge. Lower bounds on $p^*[A,\cF]$ can then be found by any feasible solution of~\eqref{eq:program_quantum} but upper bounds are harder to find (basically because the optimization in~\eqref{eq:program_quantum} is over Hilbert spaces of unbounded dimension). Building on the works of Navascu{\'e}s, Pironio and Ac{\'\i}n~\cite{Navascues07,NPA08} and Doherty, Liang, Toner and Wehner~\cite{Doherty:2008tm} in the context of games, Pironio, Navascu{\'e}s and Ac{\'\i}n~\cite{PNA10} gave asymptotically converging hierarchies of SDP relaxations for general quantum polynomial optimization (see~\cite{FNT14} for an operator algebra point of view on this hierarchy). We briefly sketch their results when applied to our more specific setting of quantum bilinear optimization as in~\eqref{eq:program_quantum}.


\subsection{Generating upper bounds}

This section mainly serves motivational purposes. As our goal is to derive semidefinite program relaxations of~\eqref{eq:program_quantum}, we first outline a simplified analysis which will lead to upper bounds. These are then identified to be equal to the levels in the hierarchy of Navascu{\'e}s, Pironio and Ac{\'\i}n. The precise connection is briefly explained in the next section. We do not provide proofs, and defer the reader to the original papers~\cite{Navascues07,NPA08} for more details. 


We first introduce some notation. Let $\Sigma_\infty$ denote the free complex *-algebra generated by the $N+M$ symbols
\begin{align}
z_1, \dots, z_N, y_1, \dots, y_M\,.
\end{align}
In other words, these are the non-commutative polynomials in the variables $z,y$. The monomials of $\Sigma_\infty$ are also called words and can be indexed by a $u=(u_1,\dots,u_\ell)$ with $u_{i}\in\{1,\dots,N+M\}$. For example, the monomial $x_u$ indexed by $u= (1,3,3,N+2)$ is defined as $x_u = z_1 z_3^2 y_2$. The degree of a monomial $x_u$, which is also called the length of the word is denoted $\ell(u)$. The unit monomial $x_{\emptyset}$ is called the empty word indexed by $\emptyset$, and has length zero. Words $x_u,x_v$ are concatenated as
\begin{align}
x_u\circ x_v\eqdef x_{u\circ v}\quad\mathrm{with}\quad u\circ v\eqdef(u_1,\dots,u_{\ell(u)},v_1,\dots,v_{\ell(v)})\,.
\end{align}
The algebra $\Sigma_\infty$ also caries a natural involution $*:\Sigma_\infty\to\Sigma_\infty$ reversing the order of words with
\begin{align}
x_{u}^* \eqdef x_{u^*} \quad\mathrm{with}\quad u^*\eqdef(u_{\ell(u)}, \dots, u_{1}) \,,
\end{align}
and being the complex conjugation for complex scalars. For a fixed integer \(n \in \Nl\), the set of words (monomials) of length up to \(n\), \(\ell(w) \leq n\), spans a vector space $\Sigma_n$ of dimension
\begin{align}\label{eq:d_n}
d(n)\eqdef\frac{(N+M)^{n+1}-1}{N+M-1}\,.
\end{align}
Now for every feasible solution $(\cH,\psi,E_\alpha,D_\beta)$ of~\eqref{eq:program_quantum}, we define the linear form
\begin{align}\label{eq:linear_form}
\tilde{\omega}: \Sigma_\infty\to\mathbb{C}\quad\mathrm{with}\quad\tilde{\omega}(u)\eqdef\bra{\psi}X_{u}\ket{\psi}\,,
\end{align}
where $X_u$ stands for the explicit representation of the word $x_u$ in terms of the operators $E_\alpha$ and $D_\beta$ for the symbols $z_\alpha$ and $y_\beta$, respectively. Next, we choose $n\in \Nl$ and consider the $d(n) \times d(n)$ matrix labeled by words $u,v$ of length $n$
\begin{align}\label{eq:npa_matrix}
\tilde{\Omega}\eqdef\sum_{u,v \in \Sigma_n}\tilde{\Omega}_{u,v}\ket{u}\bra{v}\quad\text{with entries}\quad\tilde{\Omega}_{u,v}\eqdef\bra{\psi}X_{u^*}X_v\ket{\psi}\,.
\end{align}
Here $\ket{u}\bra{v}$ refers to the matrix with all zero entries except for the entry labeled $(u,v)$ which is equal to $1$.
This matrix is positive semidefinite since it is the Gram matrix of the vectors $X_v\ket{\psi}$. Moreover, the linear constraints $f\in\cF$ generate $d(n-1) \times d(n-1)$ matrices
\begin{align}\label{eq:npa_constraingmatrix}
\tilde{\Omega}[f]\eqdef\sum_{i=0}^{N+M}f^i\sum_{u,v \in \Sigma_{n-1}}\tilde{\Omega}_{u,(i)\circ v}\ket{u}\bra{v}
\end{align}
that are positive semidefinite as well (where $(i)$ indexes words of length one: the $i$-th symbol). For the commutativity constraints between $E_{\alpha}$ and $D_{\beta}$, this can be simply captured by identifying words $u \sim v$ if $v$ can be obtained from $u$ by using commutation between $z_{\alpha}$ and $y_{\beta}$. For example, $z_{1} y_3 z_2^2 \sim z_1 z_2^2 y_3$.
Restricting in~\eqref{eq:npa_matrix} and~\eqref{eq:npa_constraingmatrix} to constraints that only involve words up to length $n$ defines a hierarchy of semi-definite program relaxations. In more detail, for any $n \geq 1$
\begin{align}\label{eq:sdp_npa}
\begin{aligned}
\tilde{\sdp}_n[A,\cF]\eqdef\;&\underset{\tilde{\Omega}^n}{\text{maximize}}
& & \sum_{\alpha,\beta} A_{\alpha,\beta} \tilde{\Omega}^n_{(\alpha), (\beta)}  \\
& \text{subject to}
& &\tilde{\Omega}^n \in \mathrm{Pos}(d(n))\\
& & &\tilde{\Omega}^n_{\emptyset, \emptyset} = 1\\
& & &\tilde{\Omega}^n_{u, v^* \circ w} = \tilde{\Omega}^n_{v \circ u, w} \quad \forall u,v,w \in \Sigma_{n} : u \circ v \in \Sigma_{n}, v \circ w \in \Sigma_n\\
& & &\tilde{\Omega}^n_{u, v} = \tilde{\Omega}^n_{u', v'} \quad \forall u,u',v,v' \in \Sigma_{n} : u \sim u', v \sim v'\\
& & &\tilde{\Omega}^n[f]\eqdef\sum_{i=0}^{N+M} f^{i} \sum_{u,v \in \Sigma_{n-1}} \tilde{\Omega}^n_{u, (i) \circ v} \ket{u} \bra{v}\in \mathrm{Pos}(d(n-1))\quad\forall f\in\cF\,,
\end{aligned}
\end{align}
where $\mathrm{Pos}(d(n))$ denotes the set of positive semidefinite matrices of size $d(n)$ as in~\eqref{eq:d_n}, and we have the total set of constraints $\cF$ as in~\eqref{eq:set_constraints}. It now turns out by comparison to~\cite{Navascues07,NPA08} that the programs~\eqref{eq:sdp_npa} match exactly the semidefinite relaxations derived by Navascu{\'e}s, Pironia and Ac{\'\i}n. 

\subsection{NPA hierarchy}

In the optimization literature the matrices \(\tilde{\Omega}^n\) appearing in the program~\eqref{eq:sdp_npa} are called moment matrices while the matrices \(\tilde{\Omega}^n[f]\) are called localizing matrices. However, the program~\eqref{eq:sdp_npa} is not derived as presented above, but by introducing dual variables of the optimization problem~\eqref{eq:program_quantum}, which then can be identified with the matrices \(\tilde{\Omega}^n\) and \(\tilde{\Omega}^n[f]\). In case the moment matrix of the optimal solution is of the form~\eqref{eq:npa_matrix}, then the optimal solution equals the value $p^*[A,\cF]$. 

Clearly the levels of the NPA hierarchy are monotonically decreasing in the sense that for any $n\in\mathbb{N}$,
\begin{align}
\tilde{\sdp}_n[A,\cF]\geq \tilde{\sdp}_{n+1}[A,\cF]\,,
\end{align}
and by the preceding discussion we also have
\begin{align}
p^*[A,\cF]\leq \tilde{\sdp}_n[A,\cF]\,.
\end{align}
The first major contribution of~\cite{Doherty:2008tm,NPA08,PNA10} was a proof that the above sequence also converges to the value of $p^*[A,\cF]$,
\begin{align}\label{eq:convergence_NPA}
p^*[A,\cF]=\lim_{n\to\infty}\tilde{\sdp}_n[A,\cF]\,.
\end{align}
under the Assumption~\ref{assmp:boundedop}. This is achieved by showing that the quadratic module can be assumed to Archimedian and an explicit construction of the Hilbert space and associated operators.\footnote{For a given set of constraints $\cF$, the quadratic module is the set of polynomials $\cP(\Sigma_\infty)$ with variables in $\Sigma_\infty$ which are of the form $\sum_i a_i^* a_i + \sum_{ij} b_{ij}^* f_i b_{ij}$ for $a_i, b_{ij} \in \cP(\Sigma_\infty)$. It is called Archimedian, if there exists a constant $C>0$ such that the polynomial $C^2 - \sum_{i=1}^l u_i^2$ is an element. Note that we again assumed that the free variables $u_1,\ldots, u_l$ are hermitian.}


The first few levels of the NPA hierarchy have been used intensively in order to understand the separation between the classical and the quantum value of two-prover games, see e.g.,~\cite{PV09}. In the following we propose an alternative SDP hierarchy. This hierarchy is not only useful for studying two-prover games but also for other problems like (entanglement-assisted) one-shot channel coding, (quantum-proof) randomness extractors, and for optimizations over the completely positive semidefinite cone.


\subsection{New Hierarchy}\label{sec:new_hierarchy}

We use a way different from~\eqref{eq:npa_constraingmatrix} for generating constraints. Instead of defining the NPA linear form $\tilde{\omega}$ as in~\eqref{eq:linear_form} we define a bilinear form $\omega:\Sigma_\infty\times\Sigma_\infty\to\mathbb{C}$ that we now describe for the case of finite-dimensional Hilbert spaces. The general case can be found in Appendix~\ref{app:missing}. Now as stated above, for finite-dimensions we can assume that the non-commutative optimization in~\eqref{eq:program_quantum} is over tensor product Hilbert spaces $\cH\otimes\cH$ with operators $E_{\alpha}\otimes\1$ and $\1\otimes D_{\beta}$ (instead of just $[E_\alpha,D_\beta]=0$ on a single space $\cH$). We start with any feasible solution $(\cH\otimes\cH,\psi,E_\alpha\otimes\1,\1\otimes D_\beta)$ where again the operators $E_\alpha$ are explicit representations of the symbols $z_\alpha$ and the operators $D_\beta$ are explicit representations of the symbols $y_\beta$. Taking the partial trace over the second space $\cH$, we denote
\begin{align}
\sigma\eqdef\tr_{\cH}\left[\proj{\psi}\right]\eqdef\sum_i\big(\idty\otimes\bra{i}\big)\proj{\psi}\big(\idty\otimes\ket{i}\big)\quad\text{and write}\quad\ket{\psi} = \left(U \otimes \sigma^{1/2}\right) \ket{\Phi}\,,
\end{align}
where $\ket{\Phi}\eqdef\sum_{i} \ket{i} \ket {i}$ for some orthonormal basis $\{\ket{i}\}$ of $\cH$ and a unitary $U$.
The objective function of the quantum bilinear optimization program~\eqref{eq:program_quantum} can then be rewritten as
\begin{align}\label{eq:introsigma}
\sum_{\alpha,\beta}A_{\alpha,\beta}\bra{\psi}E_\alpha\otimes D_\beta\ket{\psi}
&=\sum_{\alpha,\beta}A_{\alpha,\beta}\bra{\Phi}U E_\alpha U^{\dagger} \otimes(\sigma^{1/2}D_\beta\sigma^{1/2})\ket{\Phi}\\
&=\sum_{\alpha,\beta}A_{\alpha,\beta}\tr\left[\bar{U} E_\alpha^T U^{T} \sigma^{1/2}D_\beta\sigma^{1/2}\right]\,,
\end{align}
where $E^T$ denotes the transpose of the operator $E$ and $\bar{U}$ is the complex conjugate of $U$ in the basis $\{\ket{i}\}$ of $\cH$. We note that the transpose as well as the conjugation by unitary operators preserve our constraints, and hence may be just absorbed in the operators \(E_\alpha\), as we maximize over them. Hence, we get the following alternative form of~\eqref{eq:program_quantum},
\begin{align}\label{eq:alternative_program}
\begin{aligned}
p^*[A,\cG,\cK]=\;&\underset{(\cH,\sigma,E_\alpha,D_\beta)}{\text{maximize}}
& & \sum_{\alpha,\beta} A_{\alpha,\beta} \tr\left[E_\alpha\sigma^{1/2}D_\beta\sigma^{1/2}\right]  \\
& \text{subject to}
& &\sigma\succeq0,\;\tr[\sigma]=1\\
& & &g(E_1,\dots,E_N) \succeq 0\quad\forall g\in\cG\\
& & &k(D_1,\dots,D_M) \succeq 0\quad \forall k\in\cK\,,
\end{aligned}
\end{align}
under the assumption that $\cH$ is finite-dimensional (see Appendix~\ref{app:missing} for the general case). Now, for fixed $\sigma$ we define the bilinear form
\begin{align}\label{eq:bilinear_form}
\omega:\Sigma_\infty\times\Sigma_\infty\to\mathbb{C}\quad\mathrm{with}\quad\omega(u,v)\eqdef\tr\left[X_{u}\sigma^{1/2}X_v\sigma^{1/2}\right]\,.
\end{align}
Similarly as for NPA we look at the (infinite-dimensional) matrix
\begin{align}\label{eq:new_matrix}
\Omega\eqdef\sum_{u,v}\Omega_{u,v}\ket{u}\bra{v}\quad\text{with entries}\quad\Omega_{u,v}\eqdef\omega(u^{*},v)=\tr\left[X_{u^*}\sigma^{1/2}X_{v}\sigma^{1/2}\right]=\bra{\psi}X^T_{u^*}\otimes X_{v}\ket{\psi}
\end{align}
and find that it is positive semidefinite. However, the bilinear form~\eqref{eq:bilinear_form} gives us even more structure. Namely we can say that the reordered (infinite-dimensional) matrix
\begin{align}\label{eq:11_condition}
\Omega[1,1]\eqdef\sum_{s,t,u,v}\Omega_{s^*\circ t,u^* \circ v}\ket{s}\bra{t}\otimes\ket{u}\bra{v}
\end{align}
is positive semidefinite as well. To see this, take a vector $\ket{\phi} = \sum_{s,u} c_{s,u} \ket{s} \ket{u}$. Then, we have
\begin{align}
\bra{\phi} \Omega[1,1] \ket{\phi}&= \sum_{s,t,u,v} \bar{c}_{s,u} c_{t,v} \tr\left[X_{t^*} X_s \sigma^{1/2} X_{u^*} X_{v} \sigma^{1/2} \right]\notag\\
&= \sum_{s,t,u,v} \bar{c}_{s,u} c_{t,v} \tr\left[X_{s} \sigma^{1/2} X_{u^*} X_{v} \sigma^{1/2} X_{t^{*}}\right]\label{eq:semidefinite_1}\\
&= \tr\left[\left(\sum_{s,u} \bar{c}_{s,u} X_{s} \sigma^{1/2} X_{u^*} \right) \left(\sum_{t,v} c_{t,v} X_{v} \sigma^{1/2} X_{t^*} \right) \right]\label{eq:semidefinite_2}\\
&= \tr\left[  \left(\sum_{s,u} c_{s,u} X_{u} \sigma^{1/2} X_{s^*} \right)^* \left(\sum_{s,u} c_{s,u} X_{u} \sigma^{1/2} X_{s^*} \right) \right]\geq 0\,.\label{eq:semidefinite_3}
\end{align}
More generally, any pair of linear constraints $f,\hat{f}\in\cF$ from~\eqref{eq:set_constraints} generate (infinite-dimensional) matrices
\begin{align}\label{eq:new_matrix_constraint}
\Omega[f,\hat{f}]\eqdef\sum_{i,j=0}^{N+M} f^i \hat{f}^j \sum_{r,s,u,v} \Omega_{r^* \circ (i) \circ s, u^* \circ (j) \circ v} \ket{r} \bra{s} \otimes \ket{u} \bra{v}
\end{align}
that are positive semidefinite by the same argument as in~\eqref{eq:semidefinite_1}--\eqref{eq:semidefinite_3}. Now, restricting in~\eqref{eq:new_matrix} and~\eqref{eq:new_matrix_constraint} to constraints that only involve words up to length $n$ defines the $n$-th level of our new hierarchy. The variable we optimize over is now a matrix $\Omega^n$ whose rows and columns are indexed by words of length at most $n$. That is, for $n$ odd we define
\begin{align}\label{eq:sdp_odd}
\begin{aligned}
\sdp_n[A,\cF]\eqdef\;&\underset{\Omega^n}{\text{maximize}}
& & \sum_{\alpha,\beta} A_{\alpha,\beta} \Omega^n_{(\alpha),(\beta)}  \\
& \text{subject to}
& &\Omega^n \in \mathrm{Pos}(d(n))\\
& & &\Omega^n_{\emptyset, \emptyset} = 1\\
& & &\Omega^n[f,\hat{f}]\eqdef\sum_{i,j=0}^{N+M} f^i \hat{f}^j \sum_{r,s,u,v \in \Sigma_{(n-1)/2}} \Omega^n_{r^* \circ (i) \circ s, u^* \circ (j) \circ v} \ket{r} \bra{s} \otimes \ket{u} \bra{v}\\
& & &\qquad\qquad\qquad\qquad\qquad\qquad\qquad\qquad\quad\in\mathrm{Pos}(d^2(n-1)) \quad \forall f, \hat{f}\in \cF\,.
\end{aligned}
\end{align}
Note that the third constraints of the form $\Omega^n[1,f]\succeq0$ correspond to constraints $\tilde{\Omega}^n[f]\succeq0$ in the NPA hierarchy as in~\eqref{eq:sdp_npa}. For $n \geq 2$ even, we replace the last constraint in~\eqref{eq:sdp_odd} with the following constraints where $n'\eqdef(n-2)/2$:
\begin{align}\label{eq:sdp_even}
\begin{aligned}
\Omega^n[1,1]&\eqdef \sum_{\substack{r,s \in \Sigma_{n/2} \\ u,v \in \Sigma_{n/2}}} \Omega^n_{r^* \circ s, u^* \circ v} \ket{r} \bra{s} \otimes \ket{u} \bra{v}\in \mathrm{Pos}\Big(d\big(n/2\big)d\big(n/2\big)\Big) \\
\Omega^n[f]&\eqdef\sum_{i=0}^{N+M} f^i \sum_{\substack{r,s \in \Sigma_{n'} \\ u,v \in \Sigma_{n/2}}} \Omega^n_{r^* \circ (i) \circ s, u^* \circ v} \ket{r} \bra{s} \otimes \ket{u} \bra{v}\in \mathrm{Pos}\Big(d\big(n'\big)d\big(n/2\big)\Big) \quad \forall f \in \cF \\
\Omega^n[f,\hat{f}]&\eqdef\sum_{i=0}^{N+M} f^i f^j \sum_{\substack{r,s \in \Sigma_{n'} \\ u,v \in \Sigma_{n'}}} \Omega^n_{r^* \circ (i) \circ s, u^* \circ (j) \circ v} \ket{r} \bra{s} \otimes \ket{u} \bra{v}\in \mathrm{Pos}\Big(d\big(n'\big)d\big(n'\big)\Big) \quad \forall f, \hat{f} \in \cF.
\end{aligned}
\end{align}

In accordance with the literature we call the matrices \(\Omega^n\) moment matrices and the matrices \(\Omega^n[f,\hat{f}]\), \(\Omega^n[f]\) localizing matrices. Clearly the levels of this new hierarchy are monotonically decreasing in the sense that for any $n\in\mathbb{N}$,
\begin{align}
\sdp_n[A,\cF]\geq\sdp_{n+1}[A,\cF]\,.
\end{align}
We note that the SDPs we derive correspond in the special case where $\ket{\psi}$ is restricted to be a maximally entangled state on $\cH \otimes \cH$, or equivalently $\sigma$ to be maximally mixed, to the SDP relaxations proposed in~\cite{LVN14}. Such relaxations were also used for verifying experimental findings~\cite{CLBGK15}.

The following theorem summarizes the relationship between $p^*[A, \cF]$ and the sequence of SDPs $\sdp[A,\cF]$.

\begin{theorem}\label{thm:hierarchy_general}
Using the notation in this section, we have for all $n \geq 1$,
\begin{align}\label{eq:upper_bound_sdp}
p^*[A,\cF]\leq\sdp_n[A,\cF]\,.
\end{align}
Moreover, under the Assumption~\ref{assmp:boundedop} we have 
\begin{align}\label{eq:asymptotic_convergence}
p^*[A,\cF]=\lim_{n\to\infty}\sdp_n[A,\cF]\,.
\end{align}
\end{theorem}

\begin{proof}
The inequality~\eqref{eq:upper_bound_sdp} was proved above for finite-dimensional Hilbert spaces. For the general case see Appendix~\ref{app:alternative}. For~\eqref{eq:asymptotic_convergence}, a self-contained proof can be found in Appendix~\ref{app:convergence}. The convergence also follows from the convergence of the NPA hierarchy~\eqref{eq:convergence_NPA} together with Proposition~\ref{prop:comparison_hierarchies}.
\end{proof}

We now discuss the first level relaxation of our new hierarchy~\eqref{eq:sdp_odd} in more detail.


\subsection{First Level Relaxation}

For applications the first level relaxation often already gives good bounds. We find
\begin{align}\label{eq:generic-first-level}
\begin{aligned}
\sdp_1[A,\cF]=\;&\underset{\Omega^1}{\text{maximize}}
& & \sum_{\alpha,\beta} A_{\alpha,\beta} \Omega^1_{(\alpha), (\beta)}  \\
& \text{subject to}
& &\Omega^1 \in \mathrm{Pos}\left(1 + N + M\right)\\
& & &\Omega^1_{\emptyset, \emptyset} = 1\\
& & &\sum_{i,j=0}^{N+M} f^i \hat{f}^j\Omega^1_{(i),(j)}\geq 0 \quad \forall f, \hat{f}\in \cF\,.
\end{aligned}
\end{align}
Compared to this, the first level relaxation of the NPA hierarchy~\eqref{eq:sdp_npa} gives
\begin{align}\label{eq:npa-first-level}
\begin{aligned}
\tilde{\sdp}_1[A,\cF]=\;&\underset{\tilde{\Omega}^1}{\text{maximize}}
& & \sum_{\alpha,\beta} A_{\alpha,\beta} \tilde{\Omega}^1_{(\alpha), (\beta)}  \\
& \text{subject to}
& &\tilde{\Omega}^1 \in \mathrm{Pos}\left(1 + N + M\right)\\
& & &\tilde{\Omega}^1_{\emptyset, \emptyset} = 1\\
& & &\sum_{i=0}^{N+M}f^{i}\tilde{\Omega}^1_{(i),\emptyset}\geq0\quad\forall f\in\cF\,.
\end{aligned}
\end{align}
By inspection we find that~\eqref{eq:generic-first-level} has extra constraints compared to~\eqref{eq:npa-first-level}. This implies in particular that the first level of our hierarchy is never a worse approximation than the first level of the NPA hierarchy,
\begin{align}
\sdp_1[A,\cF]\leq\tilde{\sdp}_1[A,\cF]\,.
\end{align}
The extra conditions are of the form
\begin{align}
&\sum_{i,j=0}^{N+M} g^i \hat{g}^j\Omega^1_{(i),(j)}\geq 0 \quad \forall g, \hat{g}\in \cG\quad\sum_{i,j=0}^{N+M} k^i \hat{k}^j\Omega^1_{(i),(j)}\geq 0 \quad \forall k, \hat{k}\in \cK\label{eq:additional_plain}\\
&\sum_{i,j=0}^{N+M} g^i k^j\Omega^1_{(i),(j)}\geq 0 \quad \forall g\in \cG,\; \forall k\in \cK\,.\label{eq:mixed}
\end{align}
We note that in many settings the constraint~\eqref{eq:mixed} can be inferred from the second level of the NPA hierarchy and hence can be added to the first NPA level as needed when evaluating examples. The former conditions~\eqref{eq:additional_plain} however are qualitatively different from the NPA hierarchy. We will see later that for certain applications and examples the additional conditions~\eqref{eq:additional_plain} are useful (Section~\ref{sec:examples}). In the following section we compare the higher levels of the two hierarchies.

\subsection{Relations between Hierarchies}

Although a direct comparison of our new hierarchy with the NPA hierarchy is difficult (see the argument below) we can give the following connection. 

\begin{proposition}\label{prop:comparison_hierarchies}
\normalfont
As already seen in~\eqref{eq:generic-first-level} and~\eqref{eq:npa-first-level} we have
\begin{align}
\sdp_{1}[A,\cF] \leq \tilde{\sdp}_1[A,\cF]\,.
\end{align}
Moreover, for $n \geq 2$ we have
\begin{align}
\sdp_{2n}[A,\cF]\leq\tilde{\sdp}_n[A,\cF]\,.
\end{align}
\end{proposition}

\begin{proof}
Let $\Omega^{2n}$ be a feasible solution for $\sdp_{2n}[A,\cF]$ with the even level constraints as in~\eqref{eq:sdp_even}. For any $w \in \Sigma_{2n}$, let $w_z$ and $w_y$ be the subwords of $w$ containing only symbols of type $z$ and $y$ respectively. For example, if $w = z_1 y_1^2 y_2 z_3 y_1$, then $w_z = z_1 z_3$ and $w_y = y_1^2 y_2 y_1$.

We define for every $w \in \Sigma_{2n}$, the complex number $m_{w}:=\Omega^{2n}_{w_{z}, w_y}$ and let $\tilde{\Omega}^n_{u,v} = m_{u^* \circ v}$ for arbitrary words $u, v$ of length at most $n$. Because of this form, it is easily seen that $\tilde{\Omega}^n_{u, v^* \circ w} = \tilde{\Omega}^n_{v \circ u, w}$. Moreover, observe that if $w \sim w'$ then $w_{z} = w'_{z}$ as well as $w_{y} = w'_{y}$. It follows that $\tilde{\Omega}^n_{u,v} = \tilde{\Omega}^n_{u',v'}$ if $u \sim u'$ and $v \sim v'$. For the positivity constraint we write
\begin{align}
\tilde{\Omega}^n &=  \sum_{u,v \in \Sigma_{n}} \tilde{\Omega}^n_{u,v} \ket{u} \bra{v} =  \sum_{u,v \in \Sigma_{n}} \Omega^{2n}_{u^*_z \circ v_{z}, u^*_y \circ v_y} \ket{u} \bra{v} \,.
\end{align}
This matrix is a principal sub-matrix of the matrix 
\begin{align}
\sum_{s,t,u,v \in \Sigma_{n}} \Omega_{s^* \circ t, u^* \circ v} \ket{s} \bra{t} \otimes \ket{u} \bra{v} \,,
\end{align}
by only considering rows corresponding to $t$ and $s$ being words with only symbols of type $z$, and $u$ and $v$ being words with only symbols of type $y$, and also such that $\ell(s \circ u), \ell(t \circ v) \leq n$. As a result $\tilde{\Omega}^n \succeq 0$. 
For the constraints $g \in \cG$, we have
\begin{align}
\sum_{i} g^i \sum_{u,v \in \Sigma_{n-1}}  \tilde{\Omega}^n_{u, (i) \circ v} \ket{u} \bra{v} &= \sum_{i} g^i \sum_{u,v \in \Sigma_{n-1}} \Omega^{2n}_{u^*_z \circ (i) \circ v_{z}, u^*_y \circ v_y}  \ket{u} \bra{v} \,,
\end{align}
which again is a positive semidefinite matrix as it is a principal sub-matrix of $\Omega^{2n}[g]$. The positivity of $\tilde{\Omega}^n[k]$ for $k \in \cK$ is similar.
\end{proof}

This proposition implies in particular that the convergence of the new hierarchy $\sdp_n[A,\cF]$ already follows from the convergence of the NPA hierarchy $\tilde{\sdp}_{n}[A,\cF]$ (see Appendix~\ref{app:convergence} for a direct proof). We leave it as an open question if the comparison $\sdp_1[A,\cF]\leq\tilde{\sdp}_1[A,\cF]$ for the first level is special or if we might even have $\sdp_n[A,\cF]\leq\tilde{\sdp}_n[A,\cF]$ in general. We emphasize that it is unfair to directly compare the SDPs $\sdp_n[A,\cF]$ and $\tilde{\sdp}_n[A,\cF]$ as our program can have more variables. In fact, if we take into account the commutation relations in the NPA program~\eqref{eq:sdp_npa}, the variable $\tilde{\Omega}^n$ is effectively smaller than the matrix $\Omega^n$ for our new relaxation~\eqref{eq:sdp_odd}, and even more so for large $n$.


\section{Applications}\label{sec:examples}

\subsection{Two-Prover Games}\label{sec:games}

In a two-prover game, each player (or prover) gets asked a question by the referee: $q_1 \in Q_1$ for the first player and $q_2 \in Q_2$ for the second player. Each player is then asked to provide an answer $a_1 \in A_1$ and $a_2 \in A_2$. The referee, looking at the questions and answers $q_1,q_2,a_1,a_2$ decides whether the players win or lose the game according to a function $V : A_1 \times A_2 \times Q_1 \times Q_2 \to \{0,1\}$. The players may use any agreed upon protocol but they cannot communicate once they have received the questions. The fundamental quantity of interest given such a game is the largest probability of success that the players can achieve. The study of multi-prover games was introduced in~\cite{BGKL88} and has played a major role in theoretical computer science~\cite{BFL91}. It also provides a very nice interpretation for understanding non-local correlations that can be obtained by measuring an entangled state~\cite{CHTW04}. The value of a game defined by the verification predicate $V$ and a distribution $\pi$ is given by
\begin{align}\label{eq:game_classical}
\begin{aligned}
\game(V, \pi)\eqdef\;& \underset{(e,d)}{\text{maximize}}
& & \sum_{q_1,q_2} \pi(q_1,q_1) \sum_{a_1,a_2} V(a_1,a_2,q_1,q_2) e(a_1|q_1) d(a_2|q_2) \\
& \text{subject to}
& &\sum_{a_1} e(a_1|q_1) = 1 \quad \forall q_1 \in Q_1 \\
& 
& &\sum_{a_2} d(a_2|q_2) = 1 \quad \forall q_2 \in Q_2 \\
& 
& & 0 \leq e(a_1|q_1) \leq 1 \quad \forall (a_1,q_1) \in A_1 \times Q_1 \\ 
&
& & 0 \leq d(a_2|q_2) \leq 1 \quad \forall (a_2,q_2) \in A_2 \times Q_2\,.
\end{aligned}
\end{align}
In the notation of~\eqref{eq:program_classical}, we have $N = |Q_1||A_1|$, $M = |Q_2||A_2|$, $\alpha \in Q_1 \times A_1$ and $\beta \in Q_2 \times A_2$. The matrix specifying the objective function is given by 
\begin{align}
A_{(q_1,a_1), (q_2,a_2)} = \pi(q_1,q_2) V(a_1,a_2,q_1,q_2) \,.
\end{align}
The constraints functions $\cF$ are the positivity and normalization conditions.
When the players are allowed to share entanglement (of arbitrary dimension), then we define the entangled value of the game as
\begin{align}\label{eq:game_quantum}
\begin{aligned}
\game^*(V, \pi)\eqdef\;& \underset{(\cH,\psi,E,D)}{\text{maximize}}
& & \sum_{q_1,q_2} \pi(q_1,q_1) \sum_{a_1,a_2} V(a_1,a_2,q_1,q_2)\bra{\psi} E(a_1|q_1) D(a_2|q_2) \ket{\psi} \\
& \text{subject to}
& &[E(a_1|q_1), D(a_2|q_2)] = 0, \; \forall a_1,a_2,q_1,q_2 \in A_1 \times A_2 \times Q_1 \times Q_2 \\
&
& &\sum_{a_1} E(a_1|q_1) = \id_{\cH} \quad \forall q_1 \in Q_1 \\
& 
& &\sum_{a_2} D(a_2|q_2) = \id_{\cH} \quad \forall q_2 \in Q_2 \\
& 
& & 0 \preceq E(a_1|q_1) \preceq \id_{\cH} \quad \forall (a_1,q_1) \in A_1 \times Q_1 \\ 
&
& & 0 \preceq D(a_2|q_2) \preceq \id_{\cH} \quad \forall (a_2,q_2) \in A_2 \times Q_2\,.
\end{aligned}
\end{align}
Using the procedure described in Section \ref{sec:quantum}, we can define a sequence of SDPs $\omega^{\sdp_n}(V,\pi)$ that are upper bounds on $\omega^*(V, \pi)$. In particular, for $n=1$, the SDP reads
\begin{align}\label{eq:game_sdp}
\begin{aligned}
\omega^{\sdp_1}(V,\pi)\eqdef\;& \underset{\Omega^1}{\text{maximize}}
& & \sum_{q_1,q_2} \pi(q_1,q_1) \sum_{a_1,a_2} V(a_1,a_2,q_1,q_2) \Omega_{(q_1,a_1),(q_2,a_2)} \\
& \text{subject to}
& &\Omega^1 \in \textrm{Pos}(1+|Q_1||A_1| + |Q_2||A_2|) \\
&
& &\Omega^1_{\emptyset,\emptyset} = 1\\
&
& &\sum_{a_1} \Omega^1_{(q_1,a_1),u} = \Omega^1_{\emptyset, u} \quad \forall q_1 \in Q_1, u \in \Sigma_1 \\
& 
& &\sum_{a_2} \Omega^1_{(q_2,a_2),u} = \Omega^1_{\emptyset, u} \quad \forall q_2 \in Q_2, u \in \Sigma_1 \\
& 
& & \Omega^1_{u,v} \geq 0 \quad \forall u,v \in \Sigma_{1}\,.
\end{aligned}
\end{align}
We have that the boundedness condition from Assumption~\ref{assmp:boundedop} is fulfilled by the last two constraints in~\eqref{eq:game_quantum}. Compared to the first level of the NPA hierarchy, the additional constraint is the last one, namely that all the matrix entries are non-negative. Note that for the special case of two-prover games the NPA hierarchy would explicitly encode the fact that we can assume that the operators $E(a_1|q_1)$ and $D(a_2|q_2)$ define projective measurements~\cite{NPA08}. This is done by adding some relations in the algebra $\Sigma_{\infty}$: one would add the relation,\footnote{We could easily add this property as well, but we choose not to do it to simplify the exposition.}
\begin{align}
(q_i,a_i) \circ (q_i, a'_i) = \delta_{a_i=a'_i} (q_i, a_i)\quad\mathrm{for}\quad i \in \{1,2\}\,,
\end{align}
and this decreases the number of words to be considered. Using this property together with the second level of the NPA hierarchy, one could then add to the first level of NPA the constraint that the off-diagonal blocks of the matrix $\tilde{\Omega}^1$ only have non-negative elements:
\begin{align}
\tilde{\Omega}_{(q_1,a_1),(q_2,a_2)} \geq 0\quad\text{for all}\quad(q_1,a_1) \in Q_1 \times A_1\quad\text{and}\quad(q_2,a_2) \in Q_2 \times A_2\,.
\end{align}
The SDP with these non-negativity constraints for the off-diagonal blocks also appeared in the context of studying unique games in~\cite{KRT07} (see also~\cite{JP11} for a discussion of various SDP relaxations). The additional constraint in our SDP is that all the entries of the matrix $\Omega^1$ are required to be non-negative.\\
  
\textbf{Independent work:} Very recently and independently of our work, the preprint~\cite{SV15} appeared showing (among other things) that in the case of games, the first level of the NPA hierarchy can be strengthened by including the constraint that the matrix elements are non-negative. This strengthening corresponds to $\omega^{\sdp_1}$ as in~\eqref{eq:game_sdp}.


\subsection{Noisy Channel Coding}\label{sec:channel}

Let us recall the setup of channel coding from the introduction. We have a channel mapping an element from the set $X$ to an element of the set $Y$ according to probabilities given by $W_{X\to Y}(y|x)$. The objective is to determine the maximum success probability for transmitting $k$ bits of information using this channel. The classical version of the problem is described in~\eqref{eq:channel_classical}. In the notation of~\eqref{eq:program_classical}, we have $N = 2^k |X|$, $M = 2^k |Y|$, $\alpha \in [2^k] \times X$ and $\beta \in [2^k] \times Y$. The matrix specifying the objective function is given by 
\begin{align}
A_{(i,x), (j,y)} = \delta_{i=j} W_{X\to Y}(y|x) \,.
\end{align}
The constraints functions $\cF$ are the positivity and normalization conditions. Explicitly writing the first level SDP from~\eqref{eq:generic-first-level} with some easy simplifications, we get
\begin{align}\label{eq:channel_coding}
\begin{aligned}
\channel^{\sdp_1}(W,k)\eqdef\;&\underset{\Omega^1}{\text{maximize}}
& & \frac{1}{2^k} \sum_{x,y,i} W_{X \to Y}(y|x) \Omega^1_{(i,x), (i,y)}  \\
& \text{subject to}
& &\Omega^1 \in \mathrm{Pos}(1+k|X| + k|Y|)\\
& & &\Omega^1_{\emptyset,\emptyset} = 1\\
& & &\sum_{x} \Omega^1_{w, (i,x)} =  \Omega^1_{w, \emptyset} \quad \forall i \in \left[2^k\right], w \in \Sigma_{1}\\
& & &\sum_{i} \Omega^1_{w, (i,y)} = \Omega^1_{w, \emptyset} \quad \forall y \in Y, w \in \Sigma_{1} \\
& & &\Omega^1_{u, v} \geq 0 \quad \forall u,v \in \Sigma_{1}\,.
\end{aligned}
\end{align}
Again, the additional constraint compared to the NPA hierarchy is the last one, namely the fact that all the entries of $\Omega^1$ are non-negative. Using this condition, we see that we have the desirable property that for any valid channel $W$ and any $k$,
\begin{align}
\channel^{\sdp_1}(W,k) &\leq \frac{1}{2^k} \sum_{x,y,i}W_{X \to Y}(y|x)\Omega^1_{(i,x), \emptyset} =\frac{1}{2^k} \sum_{i,x}\Omega^1_{(i,x), \emptyset}=\frac{1}{2^k} \sum_{i,x}\bar{\Omega}^1_{\emptyset,(i,x)}=\frac{1}{2^k} \sum_{i} \bar{\Omega}^{1}_{\emptyset, \emptyset} = 1 \,,
\end{align}
where we have used that the matrix $\Omega^1$ is hermitian (which is implied by $\Omega^1\succeq0$). Now, as a concrete example for which the classical and the quantum success probabilities are different we mention the following setup from~\cite{PLMKR11}. The objective is to send $k=1$ bit over the noisy channel $Z_{X \to Y}(y | x)$ represented by the input-output matrix
\begin{align}\label{eq:channel_example}
\left(\begin{smallmatrix}
1/3 & 1/3 & 0 & 0\\
0 & 0 & 1/3 & 1/3\\
1/3 & 0 & 1/3 & 0\\
0 & 1/3 & 0 & 1/3\\
1/3 & 0 & 0 & 1/3\\
0 & 1/3 & 1/3 & 0
\end{smallmatrix}\right)\,.
\end{align}
It is shown in~\cite{PLMKR11} that for this channel the classical and quantum success probability as in~\eqref{eq:channel_classical} and~\eqref{eq:channel_quantum} respectively are separated as,
\begin{align}
\channel^*(Z, 1)\geq\frac{2+2^{-1/2}}{3}\approx0.902>0.833\approx\frac{5}{6}=\channel(Z, 1)\,.
\end{align}
Moreover, it was shown in~\cite{PhysRevA.87.062301,Williams11} that the above lower bound for $\channel^*(W, 1)$ is optimal as long as we restrict the optimization in~\eqref{eq:channel_quantum} to two dimensional Hilbert spaces.

Implementing our first level SDP relaxation~\eqref{eq:channel_coding} using CVX for MATLAB~\cite{cvx,gb08} gives the first non-trivial upper bound for the general optimization~\eqref{eq:channel_quantum} leading to,\footnote{The code is available at \url{http://www.omarfawzi.info}.}
\begin{align}
\channel^{\sdp_1}(Z,1) \approx 0.908\geq\channel^*(Z, 1)\geq0.902\,.
\end{align}
We note that the first level NPA relaxation as in~\eqref{eq:npa-first-level} only gives the trivial upper bound of one. This is the case even when adding the constraint that the off-diagonal elements of the matrix $\tilde{\Omega}^{1}_{(i,x),(j,y)}$ are non-negative.\footnote{Another upper bound, the so-called non-signaling success probability of the channel~\eqref{eq:channel_example}, is one as well (see~\cite{PLMKR11,PhysRevA.87.062301,Williams11} for details).} In Appendix~\ref{app:coding}, we show that the bound given by the $\channel^{\sdp_1}(Z,1)$ is in fact achievable with four dimensional entanglement-assistance:
\begin{align}
\channel^{*}(Z, 1) \geq \frac{1}{2} + \frac{1}{\sqrt{6}} \approx 0.908\,.
\end{align}\\

\textbf{Subsequent work:} After this work was posted, a limit on the maximum advantage that can be obtained by using entanglement-assistance was proved in~\cite{BF15}. More precisely, we have that for any channel $W$ and sending $k$ bits of information, 
\begin{align}
\channel(W,k) \geq (1-e^{-1})\channel^{*}(W,k)\,.
\end{align}


\subsection{Randomness Extractors}\label{sec:randomness}

A randomness extractor is defined by a set of functions
\begin{align}
\mathsf{Ext}:=\Big\{f_s:\left[2^n\right] \to \left[2^m\right]\Big\}_{s \in \left[2^d\right]}
\end{align}
mapping bit strings of length \(n\) to shorter ones of length \(m\); see~\cite{Vadhan11} for a survey. As the name suggests, the goal is to extract (almost) perfect randomness from a weaker source of randomness. That is, given some distribution over bit strings of length \(n\), by applying one of the functions chosen uniformly at random, we want to obtain a distribution close to the uniform one (in the total variation distance). The requirement is that the initial distribution contains enough randomness as measured using the min-entropy which is equal to minus the logarithm of the maximal entry of the probability distribution. In order for this procedure to work for all sources satisfying the min-entropy constraint, it can be shown that the minimal size of the seed \(d\) is logarithmic in \(n\)~\cite{Vadhan11}. Since the total variation distance between two distributions can itself be written as an optimization over test functions, the performance of a given extractor \(\mathsf{Ext}\) can be cast as a bilinear optimization program. The objective function in the general program~\eqref{eq:program_classical} is chosen to be indexed by elements \(i \in \left[2^n\right]\) and pairs \((s,j) \in \left[2^{d+m}\right]\),
\begin{align}
A_{i,(s,j)}:=\frac{1}{2^d} \delta_{f_s(i)=j} - \frac{1}{2^{d+m}}\,. 
\end{align}
The constraints are the positivity and normalization of the input distribution \(z_i\), as well as the min-entropy requirement, and the restriction to test functions as given by positive numbers \(y_{(s,j)}\). We arrive at
\begin{align}\label{eq:classical_extractor}
\begin{aligned}
\mathrm{Err}(\mathsf{Ext},k)\eqdef\;&\underset{(z_i,y_{(s,j)})}{\text{maximize}}
& & \frac{1}{2^d} \sum_{i,(s,j)} \left[ \delta_{f_s(i)=j} - \frac{1}{2^{m}}\right] z_i y_{(s,j)}  \\
& \text{subject to}
& & 0\leq z_i \leq 2^{-k}\quad\forall i\in\left[2^n\right]\\
& & &\sum_i z_i =1\\
& & & 0\leq y_{(s,j)} \leq 1\quad\forall (s,j) \in \left[2^{d+m}\right]\,.
\end{aligned}
\end{align}
Here, the parameter \(k\) measures the amount of initial min-entropy. As discussed before, the constraints on the positive numbers \(y_{(s,j)}\) just ensure that it is a test function, and hence the program becomes
\begin{align}\label{eq:extractor_1norm}
\begin{aligned}
\mathrm{Err}(\mathsf{Ext},k)=\;&\underset{z_i}{\text{maximize}}
& & \frac{1}{2}\cdot\frac{1}{2^d} \sum_{(s,j)} \left|\sum_i\delta_{f_s(i)=j} z_i - \frac{1}{2^{m}} \right|  \\
& \text{subject to}
& & 0\leq z_i \leq 2^{-k}\quad\forall i\in\left[2^n\right]\\
& & &\sum_i z_i =1\,,
\end{aligned}
\end{align}
the total variation distance of the output distribution to the uniform distribution on \(m\) bits. The average over the choice of the seed value \(s\) outside of the absolute value ensures that the closeness to the uniform distribution holds even conditioned on the seed. We also call
\begin{align}
C(\mathsf{Ext},k):=\mathrm{Err}(\mathsf{Ext},k)
\end{align}
the classical value of $\mathsf{Ext}$. We can now apply our general quantization procedure to~\eqref{eq:classical_extractor}. Assuming for simplicity that the underlying Hilbert space is of finite-dimensions and repeating the steps~\eqref{eq:introsigma} -~\eqref{eq:alternative_program}, we arrive at the program (for the general case see again Appendix~\ref{app:alternative}),
\begin{align}\label{eq:quantum_extractor}
\begin{aligned}
\mathrm{Err}^*(\mathsf{Ext},k)\eqdef\;&\underset{(\sigma,E_i,D_{(s,j)})}{\text{maximize}}
& & \frac{1}{2^d} \sum_{i,(s,j)} \left[ \delta_{f_s(i)=j} - \frac{1}{2^{m}}\right] \tr\left[E_i \sigma^{1/2} D_{(s,j)} \sigma^{1/2}\right] \\
& \text{subject to}
& & \sigma\succeq0,\;\tr[\sigma] = 1\\
& & &0\preceq E_i \preceq 2^{-k} \1\quad\forall i\in\left[2^n\right]\\
& & &\sum_i E_i =\1\\
& & & 0 \preceq D_{(s,j)} \preceq \1\quad\forall (s,j) \in \left[2^{d+m}\right]\,.
\end{aligned}
\end{align}
Setting $\sigma_i\eqdef\sigma^{1/2}E_i\sigma^{1/2}$ and again by the duality of the 1-norm to the \(\infty\)-norm we can rewrite the program as
\begin{align}\label{eq:q_extractor}
\begin{aligned}
\mathrm{Err}^*(\mathsf{Ext},k)=\;&\underset{\sigma_i}{\text{maximize}}
& & \frac{1}{2}\cdot\frac{1}{2^d} \sum_{(s,j)} \left\Vert\sum_i\left[\delta_{f_s(i)=j}-\frac{1}{2^{m}}\right]\sigma_{i} \right\Vert_1 \\
& \text{subject to}
& & 0\preceq \sigma_i \preceq 2^{-k} \sum_i\sigma_i\quad\forall i\in\left[2^n\right]\\
& & &\sum_i\tr[\sigma_i] = 1\,.
\end{aligned}
\end{align}
From this we define the normalized classical-quantum state
\begin{align}\label{eq:cq_worst}
\sigma\eqdef\sum_i \ket{i}\bra{i} \otimes \sigma_i\quad\mathrm{satisfying}\quad\sigma \preceq 2^{-k}\cdot\1 \otimes \left(\sum_i\sigma_i\right)\,,
\end{align}
and hence the objective function in~\eqref{eq:q_extractor} corresponds to the total variation distance of the output to a quantum state that is of the form uniform distribution on \(m\) bits tensor the reduced state on the quantum system. This means that an adversary cannot tell the output apart from the uniform distribution even when having access to the quantum system as well as the value of the seed. Here, the inequality condition in~\eqref{eq:cq_worst} defines the worst case quantum conditional min-entropy that is, e.g., discussed in~\cite[Appendix B]{Tomamichel11}. However, in the literature the average case quantum conditional min-entropy is more commonly used (as discussed in~\cite{Renner05}). This gives rise to the following so-called quantum value of $\mathsf{Ext}$,
\begin{align}
\begin{aligned}
Q(\mathsf{Ext},k)\eqdef\;&\underset{\left(\sigma_i,\omega\right)}{\text{maximize}}
& & \frac{1}{2}\cdot\frac{1}{2^d} \sum_{i,(s,j)} \left\Vert\sum_i\left[\delta_{f_s(i)=j}-\frac{1}{2^{m}}\right]\sigma_{i}\right\Vert_1 \\
& \text{subject to}
& & \omega\succeq0,\;\tr[\omega] = 1\\
& & &0\preceq \sigma_i \preceq 2^{-k} \omega\quad\forall i\in\left[2^n\right]\\
& & &\sum_i\tr[\sigma_i] = 1\,.
\end{aligned}
\end{align}
However, it follows from the equivalence of the worst case and average case quantum conditional min-entropy~\cite[Lemma 20]{Tomamichel11} that there cannot be a large gap between $\mathrm{Err}^*$ and $Q$.
\begin{proposition}
For $\eps>0$ we have
\begin{align}
Q(\mathsf{Ext},k)\leq\mathrm{Err}^*\left(\mathsf{Ext},k-\log\left(1/\eps^2+1\right)\right)+\eps\,.
\end{align}
\end{proposition}

We conclude that $\mathrm{Err}^*(\mathsf{Ext},k)$ captures to what extent $\mathsf{Ext}$ is a quantum-proof extractor. Hence, this property can be tested by our SDP hierarchy~\eqref{eq:sdp_odd}. We give the full first level $\mathrm{Err}^{\sdp_1}(\mathsf{Ext},k)$ as in~\eqref{eq:generic-first-level} in Appendix~\ref{app:extractor}. For our purposes, however, it will be sufficient to work with the following simplified upper bound $\mathrm{Err}^{\overline{\sdp}_1}(\mathsf{Ext},k)\geq\mathrm{Err}^{\sdp_1}(\mathsf{Ext},k)$ that ignores some of the constraints:
\begin{align}\label{eq:sdp1_extractor}
\begin{aligned}
\mathrm{Err}^{\overline{\sdp}_1}(\mathsf{Ext},k)\eqdef\;&\underset{\Omega^1}{\text{maximize}}
& & \frac{1}{2^d} \sum_{i,(s,j)} \left[ \delta_{f_s(i)=j} - \frac{1}{2^{m}}\right] \Omega^1_{(i),(s,j)} \\
& \text{subject to}
& & \Omega^1  \in \mathrm{Pos}(1+2^n + 2^{d+m})\\
& & &\Omega^1_{w,w'}\geq0\quad\forall w,w'\in \Sigma_1\\
& & &\Omega^1_{\emptyset,\emptyset} = 1,\;\Omega^1_{\emptyset,w}=\sum_{i}\Omega^1_{(i),w}\quad\forall w\in \Sigma_1\\
& & & 2^{-k} \Omega^1_{\emptyset,w}\geq\Omega^1_{(i),w}\quad \forall i\in \left[2^n\right],\;\forall w\in\Sigma_1\\
& & &\Omega^1_{\emptyset,w}\geq\Omega^1_{(s,j),w}\quad \forall (s,j)\in \left[2^{d+m}\right],\;\forall w\in\Sigma_1\,,
\end{aligned}
\end{align}
where again some of the positivity constraints on the matrix elements are new as compared to the NPA hierarchy. We emphasize that these conditions are important to obtain the following bounds on the gap between \(\mathrm{Err}(\mathsf{Ext},k)\) and \(\mathrm{Err}^{\overline{\sdp}_1}(\mathsf{Ext},k)\), which then also give an upper estimate for the error of the quantum-proof extractor~\eqref{eq:quantum_extractor}.

\begin{theorem}\label{thm:extractor}
We have that
\begin{align}\label{eq:extractor_thm1}
\mathrm{Err}^{\overline{\sdp}_1}(\mathsf{Ext},k) \leq \sqrt{2}\sqrt{2^m} \sqrt{\mathrm{Err}(\mathsf{Ext},k)}\,,
\end{align}
as well as
\begin{align}\label{eq:extractor_thm2}
\mathrm{Err}^{\overline{\sdp}_1}(\mathsf{Ext},k)\leq6\,K_G\,2^{n-k}\,\mathrm{Err}(\mathsf{Ext},k-1)\,
\end{align}
where $K_G$ denotes Grothendieck's constant.
\end{theorem}

The proof is based on ideas from~\cite[Theorem 5]{berta15} and we present it in full detail in Appendix~\ref{app:extractor}. We remark that compared to the relaxation in~\cite[Theorem 4]{berta15}, the SDP relaxation~\eqref{eq:sdp1_extractor} has some new and different constraints. The additional constraints are introduced by the sub-matrices where one variable is equal to the empty word $\emptyset$. Using these additional constraints we have the desirable property that the first level SDP relaxation $\mathrm{Err}^{\sdp_1}(\mathsf{Ext},k)$ is always bounded by one,\footnote{Both $\mathrm{Err}^*(\mathsf{Ext},k)$ and $Q(\mathsf{Ext},k)$ are always bounded by one whereas the relaxation in~\cite[Theorem 4]{berta15} can get arbitrarily large in general.}
\begin{align}
\mathrm{Err}^{\overline{\sdp}_1}(\mathsf{Ext},k)\leq \frac{1}{2^d} \sum_{i,(s,j) : f_s(i) = j} \left(1 - \frac{1}{2^m}\right) \Omega^1_{(i), (s,j)}\leq \frac{1}{2^d} \sum_{i,s} \left(1 - \frac{1}{2^m}\right)\Omega^{1}_{(i),\emptyset} = 1 - \frac{1}{2^m}\,. 
\end{align}
This implies that the argument in~\cite[Theorem 8]{berta15} showing a large gap between the SDP value and the quantum value does not apply for $\mathrm{Err}^{\overline{\sdp}_1}(\mathsf{Ext},k)$. We leave it as an open question whether there can be a large gap between $\mathrm{Err}^{\overline{\sdp}_1}(\mathsf{Ext},k)$ or $\mathrm{Err}^{\sdp_1}(\mathsf{Ext},k)$ and $\mathrm{Err}^*(\mathsf{Ext},k)$.\footnote{Some more results on how to extend the argument in~\cite[Theorem 8]{berta15} to other SDP relaxations can be found in~\cite{Fer15}.}

Finally, we point out that using ideas similar to the ones presented in this section, one can also construct a hierarchy for more general objects called quantum-proof randomness condensers~\cite{Vadhan11,FullVersion}. It would be interesting to explore in more detail the applications of these relaxations to condensers.


\subsection{Optimization over the cone \(\mathcal{CS}_+\)}

Here we show that one can use the hierarchy introduced in Section~\ref{sec:quantum} to give a SDP hierarchy of outer approximation for the cone $\mathcal{CS}^N_+$ defined in~\cite{LP13},
\begin{align}
  \mathcal{CS}^N_+ \eqdef \Big\{ \Gamma \in \mathrm{Pos}(N) \, : \, \Gamma_{\alpha, \beta} = \tr[ X_\alpha X_\beta ] \text{ with } X_1, \dots, X_N \in \mathrm{Pos}(d) \text{ for some } d \in \mathbb{N} \Big\}\,.
\end{align}
A typical program considered by Burgdorf, Laurent and Piovesan~\cite{BLP15} now reads as follows:
\begin{align}\label{eq:optimcsppluscone}
\begin{aligned}
p^{\mathcal{CS}_+}\left[A, \{F^i\}_i\right]\eqdef \; & \underset{\Lambda}{\text{maximize}} & & \sum_{\alpha,\beta} A_{\alpha,\beta} \Lambda_{\alpha,\beta}\\
& \text{subject to} 
& & \Lambda \in \mathcal{CS}^N_{+}\\
& & & \sum_{\alpha,\beta} F^i_{\alpha \beta} \Lambda_{\alpha,\beta} = G^i \quad \forall i \,.
\end{aligned}
\end{align}
Here again, $A_{\alpha, \beta}$ are real numbers specifying the objective function, and the real numbers $F^i_{\alpha, \beta}$, $G^i$ specify additional equality constraints. Specific instances include the quantum versions of stability and chromatic numbers for graphs; see e.g.,~\cite{CMNSW07,BLP15}. Note that as we do not distinguish between two types of variables here, we use $N$ instead of $N+M$ for the number of variables. As in Assumption~\ref{assmp:boundedop}, we assume that the constraints $ \sum_{\alpha,\beta} F^i_{\alpha \beta} \Lambda_{\alpha,\beta} = G^i$ are such that they imply $X_\alpha \preceq C\idty_d$ for some constant. For all applications we know of, this is satisfied.  


The above optimization problem is closely related to the tracial moment problem, tracial optimization of non-commutative polynomials as studied extensively by Burgdorf, Cafuta, Klep, and Povh~\cite{BCKP13,BK12,KP16}. In particular, Klep and Povh~\cite{KP16} studied the optimization problem of minimizing the trace of a polynomial in non-commutative variables under further positivity constraints and derived a convergent SDP hierarchy. In what is next, we describe how our general approach can be used to derive a new hierarchy especially suited for quadratic polynomials and thus for optimization over $\mathcal{CS}^N_+$.

Following the procedure given in Section \ref{sec:quantum}, the $n$-th level SDP relaxation is given by optimizing over a positive semidefinite matrix $\Omega^n$ whose rows and columns are indexed by words of length up to $n$ on the alphabet $\{1, \dots, N\}$. These words span a the complex linear subspace of \(\Sigma_\infty\) which we denote by $\Sigma_{n}$. The entries $\Omega_{(\alpha),(\beta)}$ corresponding to words of length $1$ are the candidate entries for $\Lambda_{\alpha,\beta}$ in the program~\eqref{eq:optimcsppluscone}. The fact that $\Lambda \in \mathcal{CS}^N_+$ allows us to add additional constraints as described in~\eqref{eq:new_matrix_constraint}. When $n$ is odd and writing
\begin{align}
\delta = \frac{N^{n+1} - 1}{N-1}\quad\mathrm{and}\quad\delta' = \left(\frac{N^{(n-1)/2+1} -1}{N-1}\right)^2\,,
\end{align}
we find
\begin{align}\label{eq:sdp_cspplus}
\begin{aligned}
\sdp_n[A, \{F^i\}_i] \eqdef\; &\underset{\Omega^n}{\text{maximize}} & & \sum_{\alpha,\beta} A_{\alpha,\beta} \Omega^n_{(\alpha),(\beta)}\\
& \text{subject to}
& & \Omega^n \in \mathrm{Pos}\left(\delta\right)\\
& & & \Omega^n_{\emptyset, \emptyset} = 1\\
& & &\sum_{r,s,u,v \in \Sigma_{(n-1)/2}} \Omega^n_{r^* \circ (\alpha) \circ s, u^* \circ (\beta) \circ v} \ket{r} \bra{s} \otimes \ket{u} \bra{v} \in \mathrm{Pos}(\delta') \quad \forall \alpha,\beta \in [N]\\
& & & \Omega^n_{(\alpha) \circ u, v \circ (\beta)} = \Omega^n_{u \circ (\beta), (\alpha) \circ u}  \quad \forall \alpha,\beta \in [N], u, v \in \Sigma_{n-1}\\
& & & \sum_{\alpha,\beta} F^i_{\alpha \beta} \Omega^n_{(\alpha),(\beta)} = G^i \Omega^n_{\emptyset, \emptyset} \,.
\end{aligned}
\end{align}
Recall that $\circ$ denotes the concatenation of words and $(\alpha)$ refers to a word of length $1$ with the symbol $\alpha$. Note that $n=1$ corresponds to optimizing over the doubly non-negative cone. The way we constructed $\sdp_n[A, \{F^i\}_i]$ as a relaxation of $p[A, \{F^i\}_i]$ is similar to what we did in previous sections. Let $\Lambda \in \mathcal{CS}_+^N$, then there exists positive semidefinite matrices $X'_1, \dots, X'_{N} \in \mathrm{Pos}(d)$ such that $\Gamma_{\alpha, \beta} = \tr[X'_{\alpha} X'_{\beta}]$. First, let us write $X_{\alpha} = \sqrt{d} X'_{\alpha}$ and for any word $u \in \Sigma_{n}$, $X_{u}$ as the product of the matrices corresponding to its symbols: $X_{u} = X_{u_1} \cdots X_{u_n}$ with $X_{\emptyset} = \1$. Recalling that $u^*$ is the word $u$ inverted, we define
\begin{align}
\Omega^n_{u,v}:=\tr\left[\frac{\1}{d} \cdot X_{u^*} X_{v}\right]\,.
\end{align}
First, as $\Omega^n$ is the (scaled) Gram matrix of the family $\{X_{u} : u \in \Sigma_{n}\}$, it is positive semidefinite. Also $\Omega^n_{1,1} = \tr\left[\1/d\right] = 1$. Moreover, for a vector $\ket{\phi} = \sum_{r,u} c_{r,u} \ket{r} \ket{u}$, we have
\begin{align}
&\bra{\phi} \sum_{r,s,u,v \in \Sigma_{(n-1)/2}} d \cdot \Omega^n_{r^* \circ (\alpha) \circ s, u^* \circ (\beta) \circ v} \ket{r} \bra{s} \otimes \ket{u} \langle v|\phi\rangle\notag\\
&= \sum_{r,s,u,v} \bar{c}_{r,u} c_{s,v} \tr\Big[X_{s^*} X_{\alpha} X_{r} X_{u^*} X_{\beta} X_{v}\Big] \\
&=  \tr\left[\left(\sum_{s,v} c_{s,v} X_{v} X_{s^*} \right) X_{\alpha} \left(\sum_{r,u} \bar{c}_{r,u} X_{r} X_{u^*} \right) X_{\beta}\right] \\
&=  \tr\left[\left(\sum_{s,v} c_{s,v} X_{v} X_{s^*} \right) X_{\alpha} \left(\sum_{s,v} c_{s,v} X_{v} X_{s^*} \right)^* X_{\beta}\right] \geq 0 \,.
\end{align}
The constraint
\begin{align}
\Omega^n_{(\alpha) \circ u, v \circ (\beta)} = \Omega^n_{u \circ (\beta), (\alpha) \circ u}
\end{align}
corresponds to the cyclicity of the trace,
\begin{align}
d \cdot \Omega^n_{(\alpha) \circ u, v \circ (\beta)} = \tr\Big[X_{u^{*}} X_{\alpha} X_{v} X_{\beta}\Big] = \tr\Big[X_{\beta} X_{u^{*}} X_{\alpha} X_v \Big] = d \cdot \Omega^n_{u \circ (\beta), (\alpha) \circ v}\,.
\end{align}
Note that such a constraint did not appear in our other examples as we were optimizing over the state involved in defining $\Omega^n$. In this example, we want to fix the state to be maximally mixed, $\1/d$, and this is reflected in the cyclicity condition. We can also define the SDPs for even $n$ similarly as in~\eqref{eq:sdp_even}.

We implemented the SDP relaxations to test whether a given matrix $K$ is in $\mathcal{CS}_+$. In~\cite{FW14,LP13} it was shown that the matrix
\begin{align}\label{eq:matrix_k}
K \eqdef \left(\begin{smallmatrix}
4 & 0 & 2 & 2 & 0\\
0 & 4 & 0 & 2 & 2\\
2 & 0 & 4 & 0 & 3\\
2 & 2 & 0 & 4 & 0\\
0 & 2 & 3 & 0 & 4
\end{smallmatrix}\right)
\end{align}
is not in the closure of $\mathcal{CS}_+$. Using CVX for MATLAB~\cite{cvx,gb08}, we were able to numerically certify using level $n=3$ of the hierarchy that the matrix is indeed not in the cone $\mathcal{CS}_+$.\footnote{The code is available at~\url{http://www.omarfawzi.info}.}

The convergence proof of Theorem~\ref{thm:hierarchy_general} covers the above case as well, which then raises the question how the limiting point
\begin{align}
p^*[A, \{F^i\}_i]:=\lim_{n\to\infty}\sdp_n\left[A, \{F^i\}_i\right]
\end{align}
of the programs~\eqref{eq:sdp_cspplus} can be represented. Not surprisingly, we cannot assert that it corresponds to an element in the cone \(\mathcal{CS}_+\) which asks for an underlying finite-dimensional Hilbert space. However, as shown in the Appendix~\ref{app:connes}, the assumption that Connes' embedding conjecture~\cite{Ozawa:2013hn,Connes:1976ta} has a positive answer implies that the value \(p^*[A, \{F^i\}_i]\) agrees with the program~\eqref{eq:optimcsppluscone}, or more precisely, with its value if optimized over the closure \(\overline{\mathcal{CS}}_+\) of the cone \(\mathcal{CS}_+\).\footnote{We write maximize in~\eqref{eq:optimcsppluscone}, which is consistent with the statement that the maximum is not attained (cf.~Footnote~\ref{ft:maximize}). Clearly, the limiting point of our SDP hierarchy then corresponds to the supremum, and hence to the optimization over the closure of the cone \(\mathcal{CS}_+\).}

\begin{corollary}
For any $n \geq 1$ we have
\begin{align}
p^{\mathcal{CS}_+}[A, \{F^i\}_i] \leq \sdp_n[A, \{F^i\}_i]\,.
\end{align}
Moreover, provided the Connes embedding conjecture has a positive answer~\cite{Connes:1976ta,Ozawa:2013hn}, we have
\begin{align}\label{eq:convergence_csplus}
p^{\overline{\mathcal{CS}}_+}[A, \{F^i\}_i] = p^*[A, \{F^i\}_i]\,.
\end{align}
\end{corollary}

In order to prove \eqref{eq:convergence_csplus}, we could first either relate to Klep and Povh's result~\cite{KP16} or make use of the fact that our hierarchy converges to the same value as the NPA hierarchy with added cyclicity constraints. Both approaches would imply that the state \(\tau\) on \(\Sigma_\infty\) constructed in the convergence proof is a tracial state, that is \(\tau(ab) = \tau(ba)\). However, if Connes' embedding conjecture holds, then this state can be represented as a tracial state on the ultrapower of the hyperfinite factor. Finally, Burgdorf, Laurent and Piovesan~\cite{BLP15} have shown that this implies the stated result. For the convenience of the reader, we present such an argument in Appendix~\ref{app:connescsp}.


\appendix

\section{Missing Proofs for General Hilbert Spaces}\label{app:missing}

\subsection{Upper Bounds on the Quantum Value}\label{app:alternative}

Here we show that even if we allow for general Hilbert spaces in the quantum program~\eqref{eq:program_quantum}, than SDP hierarchy~\eqref{eq:sdp_odd}-\eqref{eq:sdp_even} is still a relaxation thereof (in Section~\ref{sec:new_hierarchy} we have only shown this for finite-dimensional spaces). For this we start from the quantum program~\eqref{eq:program_quantum} and upper bound it in a more algebraic form.

Given any feasible solution $(\cH,\psi,E_\alpha,D_\beta)$ of the quantum program~\eqref{eq:program_quantum} we consider the algebra generated by the operators
\begin{align}
D_1, \dots, D_M\,,
\end{align}
acting on the Hilbert space \(\cH\), and denote its closure in operator norm by \(\mathcal{D}\). This is then a \(C^*\)-algebra and we denote the set of of hermitian functionals on  \(\mathcal{D}\) by $\mathcal{D}^*_h$ and the set of positive functionals by \(\mathcal{D}^*_+\). Now the normalized vector \(\psi\in\cH\) induces a normalized positive functional \(\sigma\in\mathcal{D}^*_+\) via
\begin{align}
\mathcal{D}\ni D\mapsto \sigma(D)\eqdef\Scp{\psi}{D\,\psi}\,.
\end{align}
Moreover, the hermitian operators \(E_\alpha\) induce positive functionals \(\rho_\alpha\in\mathcal{D}^*_h\),
\begin{align}
\mathcal{D}\ni D\mapsto \rho_\alpha(D)\eqdef\Scp{\psi}{E_\alpha D\,\psi} = \Scp{\psi}{E_\alpha^{1/2} DE_\alpha^{1/2}\,\psi}\,,
\end{align}
where the last equality follows from the commutativity constraint \([E_\alpha,D_\beta] = 0\). In order to find an upper bound on the quantum value $p^*[A,\cF]$, we consider the following optimization program over all \(C^*\)-algebras $\mathcal{D}$, 
\begin{align}\label{prog:ncbibipartitealg}
\begin{aligned}
\bar{p}[A,\cF]\eqdef\;&\underset{(\mathcal{D},\rho_\alpha,D_\beta)}{\text{maximize}}
& & \sum_{\alpha,\beta} A_{\alpha,\beta}\rho_\alpha(D_\beta)\\
& \text{subject to}
& &\rho_\alpha\in\mathcal{D}^*_h,\,\sigma\in\mathcal{D}_+^*\,\,\text{with}\,\,\sigma(\1) = 1\\
& & &g(\rho_1,\dots,\rho_N,\sigma) \succeq 0\quad \forall g\in\cG\\
& & &k(D_1,\dots,D_M) \succeq 0\quad \forall k\in\cK\,,
\end{aligned}
\end{align}
where the constraints $g\in\cG$ are now understood as
\begin{align}
g(\rho_1,\dots,\rho_N,\sigma)\eqdef g^0\sigma+\sum_{\alpha\in[N]}g^\alpha\rho_\alpha\,,
\end{align}
and positivity is read in the algebraic sense. Note that the boundedness constraints (cf.~Assumption~\ref{assmp:boundedop}) translate to
\begin{align}
\forall\alpha\in[N]:\;\;C\sigma\succeq \rho_\alpha\succeq-C\sigma\quad\mathrm{and}\quad\forall\beta\in[M]:\;\;C\1\succeq D_\beta\succeq-C\1\,.
\end{align}
Now we show that the SDP hierarchy~\eqref{eq:sdp_odd}-\eqref{eq:sdp_even} is an upper bound on the algebraic program~\eqref{prog:ncbibipartitealg}, and with that also on the quantum program~\eqref{eq:program_quantum}.

\begin{proposition}\label{prop:alternative_converse}
For any $n\in\mathbb{N}$ we have that,\footnote{We will see in Appendix~\ref{eq:channel_quantum} that even $\bar{p}[A,\cF]=p^*[A,\cF]$.}
\begin{align}\label{eq:star_bar}
\sdp_n[A,\cF]\geq\bar{p}[A,\cF]\geq p^*[A,\cF]\,.
\end{align}
\end{proposition}

\begin{proof}
The second inequality follows from the discussion above and we now prove the first inequality.

Let \(\rho_\alpha,\sigma\) be the set of functionals associated to the optimal solution \(\bar{p}[A,\cF]\). A standard GNS construction for the state \(\sigma\) gives rise to a Hilbert space \(\cH\), a dense mapping \(i:\mathcal{D} \to \cH\), a vector \(\xi = i(\1)\)  and a representation \(\pi:\mathcal{D} \to \cB(\cH)\) defined by \(\pi(x)i(a) = i(xa)\), such that
\begin{align}
\Scp{i(a)}{i(b)} = \sigma(a^* b)\quad\text{and}\quad\sigma(x) = \Scp{\xi}{\pi(x) \xi} \,.
\end{align}
For the sake of convenience, we identify \(D_\beta\) with \(\pi(D_\beta)\). By the von Neumann commutant theorem, the double commutant \(\pi(\mathcal{D})''\) of \(\pi(\mathcal{D})\) is a von Neumann algebra, denoted by \(\cM\), and the vector \(\xi\) defines a normal state on \(\cM\). Now, by~\cite[Theorem 2.2]{Woronowicz72}, there exists an anti-unitary operator \(J :\cH \to \cH\), satisfying \(J^2 = \1\), another vector \(\psi \in \cH\) (differing from \(\xi\) by at most a phase), such that for all \(Y \in \cM\)
\begin{align}
\Scp{\xi}{Y \xi}= \Scp{\psi}{Y \psi}\,,\quad J \psi = \psi\,,\quad\text{and}\quad\Scp{\psi}{Y J Y \psi} \geq 0\,.
\end{align}
Moreover, we have that \(J \cM J = \cM'\), meaning that for any operator \(X\) in the commutant of \(\cM\) there exists an element \(Y \in \cM\) such that \( J Y J = X\). By the non-commutative Radon-Nikodym derivative argument, see, e.g., \cite{Takesaki02_2}, setting
\begin{align}
h_\alpha:\cH \to \cH\,,\quad\Scp{i(a)}{h_\alpha i(b)} = \rho_\alpha(a^* b)
\end{align}
defines an operator which is positive and bounded, since
\begin{align}
0 \preceq\Scp{i(a)}{h_\alpha i(a)}=\rho_\alpha(a^* a)\preceq C \sigma(a^* a) = C\Scp{i(a)}{i(a)}\,.
\end{align}
A standard calculation also gives that \(h_\alpha \in \cM'\). Moreover, for any linear constraint \(g(\rho_1,\dots,\rho_N,\sigma) \succeq 0\) we have
\begin{align}
\Scp{i(a)}{g(h_1,\dots,h_N,\1) i(a)} &= g(\Scp{i(a)}{h_1 i(a)},\dots,\Scp{i(a)}{h_N i(a)},\Scp{i(a)}{i(a)}) \\
&= g(\rho_1,\dots,\rho_N,\sigma)(a^* a) \geq 0
\end{align}
and hence \(g(h_1,\dots,h_N,\1)\) defines a positive operator. By the previous assertions, we have that \(E_\alpha = J h_\alpha J\) is an element of \(\cM\) and likewise \(g(E_1,\dots,E_N,\1) \succeq 0\). 
  
We have all necessary ingredients at hand to define the analogue of the bilinear form $\omega:\Sigma_\infty\times\Sigma_\infty\to\mathbb{C}$ from~\eqref{eq:bilinear_form}. First, let us abbreviate for \(\gamma\in \{1,\dots,N+M\}\)
\begin{align}
Z_\gamma\eqdef\left\{
\begin{array}{lr}
X_\gamma:& k \in \{1,\dots,N\}\\
Y_{\gamma-N}:& k \in \{N+1,\dots N+M\} 
\end{array}
\right.\,,
\end{align}
and for any word \(u = (u_1,\dots,u_\ell)\), \(u_i \in \{1,\dots,N+M\}\),
\begin{align}
Z_u\eqdef Z_{u_1} \dots Z_{u_\ell} \,.
\end{align}
For any two words \(u,v\) we set
\begin{align}
\omega(u^*,v)\eqdef \Scp{\psi}{Z_v J Z_u \psi}\,,
\end{align}
and this also defines the matrix $\Omega$ as in~\eqref{eq:new_matrix}. This matrix is positive semidefinite by
\begin{align}
\sum_{u,v} \overline{\lambda_u} \lambda_v\, \omega(u^*,v) &= \sum_{u,v} \overline{\lambda_u} \lambda_v \Scp{\psi}{Z_v J Z_u \psi}= \Scp{\psi}{\sum_v \lambda_v Z_v J \sum_u \lambda_u Z_u \psi} \geq 0\,,
\end{align} 
where it is essential that \(J\) is an anti-unitary operator. Moreover, property~\eqref{eq:11_condition} is checked by
\begin{align}
&\sum_{s,u,t,v} \overline{c_{su}} c_{tv} \,\Omega_{s^* t,u^* v} = \sum_{s,u,t,v} \overline{c_{su}} c_{tv} \,\omega((s^*t)^*,u^* v)\notag\\
&\quad=\sum_{s,u,t,v} \overline{c_{su}} c_{tv} \Scp{\psi}{Z_u^* Z_v J Z_s^* Z_t \psi} = \sum_{s,u,t,v} \overline{c_{su}} c_{tv} \Scp{Z_u \psi}{J J Z_v J Z_s^* Z_t \psi}\notag\\
&\quad=\sum_{s,u,t,v} \overline{c_{su}} c_{tv} \Scp{Z_u \psi}{J Z_s^* Z_t J Z_v \psi} = \sum_{s,u,t,v} \overline{c_{su}} c_{tv} \Scp{Z_t J Z_v \psi}{Z_s J Z_u \psi} \geq 0\,,
\end{align}
which defines a positive matrix. The linear constrained assertion~\eqref{eq:new_matrix_constraint} follows in a similar way. From the previous definitions, we have that
\begin{align}
\rho_\alpha(D_\beta) = \Scp{\psi}{h_\alpha D_\beta \psi} = \Scp{\psi}{J E_\alpha J D_\beta \psi} = \Scp{\psi}{D_\beta J E_\alpha \psi} = \Omega_{(\alpha),(\beta)}\,,
\end{align}
and hence the (infinite-dimensional) matrix $\Omega$ fulfills the constraints given by any finite level $n$ as in~\eqref{eq:sdp_odd}.
\end{proof}


\subsection{Asymptotic Convergence}\label{app:convergence}

Here we show that the hierarchy~\eqref{eq:sdp_odd}-\eqref{eq:sdp_even} asymptotically converges to the quantum value~\eqref{eq:program_quantum}. The argument follows previous works~\cite{Doherty:2008tm,PNA10,NPA08}.

\begin{theorem}\label{prop:GNSforspforms}
  Let \(\sdp_n[A,\cF]\) denote the SDP hierarchy~\eqref{eq:sdp_odd}-\eqref{eq:sdp_even} of the quantum bilinear program~\eqref{eq:program_quantum}, and assume~\ref{assmp:boundedop}. Then, we have the following:
\begin{enumerate}
\item In the limit of \(n \to \infty\) the optimal solutions of the programs \(\sdp_n[A,\cF]\) converge to a finite value,
\begin{align}
\lim_{n\to\infty}\sdp_n[A,\cF]=\hat{p}[A,\cF]\,.
\end{align}
\item There exists a Hilbert space \(\cH\), a normalized vector \(\xi \in \cH\), a *-homomorphism \(\pi: \Sigma_\infty \to \cB(\cH)\) as well as a linear and positive mapping \(\vphi :\Sigma_\infty \to \cB(\cH)\) with commuting ranges (that is, \([\vphi(a),\pi(b)] = 0\) for all \(a,b \in \Sigma_\infty\)) as well as elements \(z_\alpha,y_\beta\in\Sigma_\infty\) such that
\begin{align}\label{eq:p_hat}
\hat{p}[A,\cF] = \sum_{\alpha,\beta} A_{\alpha,\beta} \Scp{\xi}{\vphi(z_\alpha) \pi(y_\beta) \xi}\,.
\end{align}
Moreover, the constraints given by the linear functions \(g\in\cG\) and $k\in\cK$ are all satisfied,
\begin{align}\label{eq:constraints_rep}
g\big(\vphi(z_1),\dots,\vphi(z_N)\big)\succeq 0,\quad\text{as well as}\quad k\big(\pi(y_1),\dots,\pi(y_M)\big)\succeq 0\,.
\end{align}
\end{enumerate}
\end{theorem}

Since the quantum bilinear program~\eqref{eq:program_quantum} is a maximization over all all expressions as on the right-hand side of~\eqref{eq:p_hat} under the constraints~\eqref{eq:constraints_rep}, it immediately follows that
\begin{align}
p^*[A,\cF]\geq\hat{p}[A,\cF]\,.
\end{align}
Now because inequality in the other direction was already established in~\eqref{eq:star_bar}, we conclude that the hierarchy~\eqref{eq:sdp_odd}-\eqref{eq:sdp_even} asymptotically converges to the quantum value~\eqref{eq:program_quantum},
\begin{align}
p^*[A,\cF]=\lim_{n\to\infty}\sdp_n[A,\cF]\,.
\end{align}
Furthermore, the optimal value $\bar{p}[A,\cF]$ of the algebraic optimization~\eqref{prog:ncbibipartitealg} also becomes equal to the quantum value
\begin{align}
p^*[A,\cF]=\bar{p}[A,\cF]\,,
\end{align}
again by~\eqref{eq:star_bar}.
 
\begin{proof}[Proof of Theorem~\ref{app:convergence}]
  We first note that due to Assumption~\ref{assmp:boundedop}, the positivity constraints provide a bound on the diagonal elements of the \(d(1) \times d(1)\) sub-matrix \(\Omega^n_{(\alpha),(\beta)}\),
\begin{align}
C^2-\Omega_{(\alpha),(\alpha)}^n\geq0\,,
\end{align}
and thus on its trace. Hence, we find that
\begin{align}
0 \leq\left|\sum_{\alpha,\beta}A_{\alpha,\beta}\Omega^n_{(\alpha),(\beta)}\right|\leq\norm{A}d(1)C^2\,.
\end{align}
Moreover, we have \(\sdp_n[A,\cF] \leq \sdp_m[A,\cF]\) for \(n \leq m\). Thus, the sequence \(\sdp_n[A,\cF]\) is monotonically decreasing and lower bounded by zero, hence converging to a finite value \(\hat{p}[A,\cF]\).

In order to proceed, we need another expression for the limiting point \(\hat{p}[A,\cF]\). More precisely, we have to examine in which way the limiting point can be seen as being specified by an infinite-dimensional matrix, capturing the constraints on all words of all possible lengths at once. For any $n\in\mathbb{N}$ we have the subspace \(\Sigma_n = \{\,a \in \Sigma_\infty \,|\,a = \sum_{w:l(w)\leq n} c_w \,w\,\}\). Furthermore, for $n$ odd we define the two families of cones
\begin{align}
\text{sym}(\Sigma_n) &\eqdef\left\{x\in \Sigma_n \otimes \Sigma_n\,\big|\,x=\sum_i \lambda_i  a_i^* \otimes a_i\,,\;a_i\in \Sigma_n\,,\lambda_i \geq 0\right\} \\
(\Sigma_n \otimes \Sigma_n)_{+} &\eqdef\left\{x \in \Sigma_n \otimes \Sigma_n\,\big|\,x=\sum_{\stackrel{f,\hat{f}\in\cF}{k,l}} \,a^*_l f a_k \otimes  b_k^* \hat{f} b_l\,,\;\;a_k,b_k \in \Sigma_{(n-1)/2}\right\}\,.
\end{align}
Let \(\Omega^n\) be a feasible point of the $n$-th level of the SDP hierarchy~\eqref{eq:sdp_odd}. By mapping a pair of words \(u,v \in \Sigma_n\) to \(\Omega^n_{u,v}\), we specify a linear functional \(\omega\) on \(\Sigma_n \otimes \Sigma_n\), and it is easily seen that the constraints on \(\Omega^n\) imply that
\begin{align}
\omega(\text{sym}(\Sigma_n) \cup (\Sigma_n \otimes \Sigma_n)_{+} ) \geq 0\,,\qquad \omega(\1) = 1 \,. 
\end{align}
For the value
\begin{align}
p'[A,\cF]\eqdef\inf\Big\{q:\,\exists\,n\,\,\text{with}\,\,q\1-\sum_{\alpha,\beta}A_{\alpha,\beta} z_\alpha\otimes y_\beta\in\text{sym}(\Sigma_n)\cup(\Sigma_n\otimes\Sigma_n)_{+}\Big\}\,,
\end{align}
we find that for a finite \(\eps > 0\) there exists an \(n\in\mathbb{N}\) such that
\begin{align}
p'[A,\cF]+\eps\geq\sum_{\alpha,\beta} A_{\alpha,\beta} \Omega^n_{(\alpha),(\beta)}\,.
\end{align}
Hence, we have \(p'[A,\cF]\geq\hat{p}[A,\cF]\). But by exploiting the Positivstellensatz of Helton and McCullough~\cite{HMcC04}, a duality argument shows (see, e.g., \cite{Doherty:2008tm}), 
\begin{align}
p'[A,\cF]= \sup\Big\{\big|\sum_{\alpha,\beta} A_{\alpha,\beta} \omega(z_\alpha \otimes y_\beta) \big|\;:\;\omega\big(\text{sym}(\Sigma_\infty) \cup (\Sigma_\infty \otimes \Sigma_\infty)_{+}\big)\geq 0,\,\omega(\1)=1\Big\}\,,
\end{align}
which then implies \(p'[A,\cF]=\hat{p}[A,\cF]\). In the following, we show how to construct a Hilbert space and associated representations, starting from \(\omega\).

As usual, the argument is based on a GNS construction, and closely follows the ideas of Woronowicz in his study of purifications for states on \(C^*\)-algebras,~\cite{Woronowicz73}. We first turn the free algebra \(\Sigma_\infty\) into a \(C^*\)-algebra, that is a norm-closed algebra such that we have \(\norm{x^* x} = \norm{x}^2\). This is achieved by defining for \(x \in \Sigma_\infty\)
\begin{align}
\norm{x} = \sup\big\{\,\norm{\pi(x)}_{\cB(\cH_{\pi})}\,:\,\pi:\Sigma_\infty \to \cB(\cH_\pi) \,\text{  a *-representation}\big\}\,.
\end{align}
Here, a *-representation is a algebraic homomorphism of \(\Sigma_\infty\) into the bounded operators on some Hilbert space \(\cH\) such that the *-involution is mapped to the usual involution. It is easily checked that this norm satisfies our requirement, and thus the topological closure of \(\Sigma_\infty\) under this norm is a \(C^*\)-algebra, which we denote by \(\cA\). For all $x_i\in\Sigma_\infty$ the Assumption~\ref{assmp:boundedop} implies  
\begin{align}
\exists C :\; C \1 \succeq x_i^2\,,
\end{align}
ensures that \(\norm{x_i} \leq \sqrt{C}\) and hence \(x_i \in \cA\) since by definition of positivity in \(\cA\) there exists \(w_i \in \cA\) with \(x_i^* x_i + w_i^* w_i = C \1\) and we have for any \(\pi:\cA \to \cB(\cH)\) and any \(\psi \in \cH, \norm{\psi} = 1\)
\begin{align}
\Scp{\psi}{\pi(x_i^2) \psi} \leq \Scp{\pi(x_i)\psi}{\pi(x_i)} + \Scp{\pi(w_i)\psi}{\pi(w_i)} = C\,.
\end{align}
We also define the opposite \(C^*\)-algebra \(\bar{\cA}\), which is as topological space equal to \(\cA\) equipped with the multiplication rule \(a \cdot b = b a\) for \(a,b \in \cA\). Following~\cite{Woronowicz72}, we denote \(a^*\) as seen as an element of \(\bar{\cA}\) by \(\bar{a}\). Then the mapping \( a \mapsto \bar{a}\) is a *-invariant, anti-linear multiplicative isometry from \(\cA\) to \(\bar{\cA}\). 

Let \(\bar{\cA} \otimes \cA\) be the maximal \(C^*\)-tensor product of \(\bar{\cA}\) and \(\cA\), see for example~\cite{Pis03} for a precise definition. On this algebra, we can define another *-invariant, anti-linear and multiplicative mapping \(j:\bar{\cA} \otimes \cA \to \bar{\cA} \otimes \cA\) satisfying \(j^2 = \id\) by setting
\begin{align}
j(\bar{a} \otimes b) = \bar{b} \otimes a \,.
\end{align}
We define a state \(s\) on \(\bar{\cA} \otimes \cA\) by setting
\begin{align}
s(\bar{a} \otimes b) = \omega(a^*,b) \,,
\end{align}
for words of finite length \(a,b\) and then extending to the closure. Normalization is immediate and positivity follows from property \(\omega((\Sigma_n \otimes \Sigma_n)_{+})\geq 0\),
\begin{align}
s((\bar{a} \otimes b)^* \bar{a} \otimes b) = s(\overline{a^* a} \otimes b^* b) = \omega(a^* a, b^* b)
\end{align}
since \(\bar{x} \cdot \bar{y} = \overline{xy}\).

Carrying out the standard GNS construction for the state \(s\) gives rise to a Hilbert space \(\cH\), a dense mapping \(i:\bar{\cA} \otimes \cA \to \cH\), a vector \(\xi = i(\1)\)  and a representation \(\pi:\bar{\cA} \otimes \cA \to \cB(\cH)\) defined by \(\pi(\bar{a} \otimes b)i(\bar{c} \otimes d) = i(\overline{ac} \otimes bd)\), such that
\begin{align}
\Scp{i(\bar{a} \otimes b)}{i(\bar{c} \otimes d)}&= \omega(c^* a, b^* d)\\
s(\bar{a} \otimes b)&= \Scp{\xi}{\pi(\bar{a} \otimes b) \xi} \,.
\end{align}
We now define an anti-linear operator \(J\) by defining it on the dense domain \(i(\bar{\cA} \otimes \cA)\) as
\begin{align}
\Scp{i(\bar{a} \otimes b)}{J i(\bar{c} \otimes d)} = \Scp{i(\bar{a} \otimes b)}{i(\bar{d} \otimes c)} \,.
\end{align}
Its adjoint equals itself, since \(\omega(a^*,b) = \overline{\omega(b^*,a)}\) due to positivity and
\begin{align}
\Scp{Ji(\bar{a} \otimes b)}{i(\bar{c} \otimes d)}=\Scp{i(\bar{b} \otimes a)}{i(\bar{c} \otimes d)}=\omega(c^* b, a^* d) = \overline{\omega(d^* c,b^* c)}=\overline{\Scp{i(\bar{a} \otimes b)}{J i(\bar{c} \otimes d)}} \,.
\end{align}
Moreover, we find that \(J^2 = \1\) and hence \(J\) can be extended to an anti-unitary involution on \(\cH\). Furthermore, we have
\begin{align}
J \pi(\bar{a} \otimes b) J i(\bar{c} \otimes d) = i(\overline{b c} \otimes a d) 
= \pi(\bar{b} \otimes a)i(\bar{c} \otimes d) = \pi(j(\bar{a} \otimes b)) i(\bar{c} \otimes d)
\end{align}
and hence \(J \pi(\bar{a} \otimes b) J = \pi(j(\bar{a} \otimes b))\). A similar calculation gives
\begin{align}
\pi(\bar{a} \otimes b) = J \pi(j(\bar{a} \otimes \1)) J \pi(\bar{\1} \otimes b)  = \pi(\bar{\1} \otimes b) J \pi(j(\bar{a} \otimes \1)) J  \,. 
\end{align}
Hence the image of the linear mapping
\begin{align}
\vphi\,:\,a \mapsto \pi(\overline{a^*} \otimes \1)=J\pi(j(\overline{a^*} \otimes \1)) J
\end{align}
is contained in the commutant of \(\pi(\1 \otimes \cA)\). Moreover, any positive element \(a^* a \in \cA\) gets mapped to
\begin{align}
\vphi(a^* a) = \pi(\bar{a}^* \bar{a} \otimes \1)=\pi(\bar{a} \otimes \1)^* \pi(\bar{a} \otimes \1)\,,
\end{align}
which is a positive operator. This proves~\eqref{eq:p_hat}. The last assertion~\eqref{eq:constraints_rep} follows similarly. Considering a linear constraint \(k(y_1,\dots,y_M)\in\cK\), we find evaluating the diagonal matrix elements of \(\pi(\1 \otimes k(y_1,\dots,y_M))\) that
\begin{align}
\Scp{i(\bar{a} \otimes b)}{\pi(\1 \otimes k(y_1,\dots,y_M) i(\bar{a} \otimes b)} &= \Scp{i(\bar{a} \otimes b)}{i(\bar{a} \otimes k(y_1,\dots,y_M)b)} \\
&= \omega(a^* a,b^* k(y_1,\dots,y_M) b) \geq 0\,.
\end{align}
Hence  \(\pi(\1 \otimes k(y_1,\dots,y_M))\) is a positive operator. A similar derivation can be carried out for the map \(\vphi\).
\end{proof}


\section{Implications of Connes' embedding conjecture}\label{app:connes}

In this appendix, we discuss the implications of a positive answer to Connes' embedding conjecture~\cite{Ozawa:2013hn,Connes:1976ta} to our hierarchy. We first give a short sketch of an argument why a positive answer to Connes' embedding conjecture implies that the optimization in the program~\eqref{eq:program_quantum} can be restricted to finite-dimensional Hilbert spaces, though it does not imply that this supremum is also achieved. In the second part of this appendix, we sketch the argument for the case of the completely positive-semidefinite cone \(\mathcal{CS}_+\). As we do not want to go into the details about Connes' embedding conjecture, its different forms and its far reaching consequences (independent of the actual answer), we refer the interested reader to the extensive reviews of Ozawa on the topic~\cite{Ozawa_survey1,Ozawa:2013hn}. 


\subsection{General case} 

In Theorem~\ref{prop:GNSforspforms}, we found that the limiting point of our SDP hierarchy can be expressed as
\begin{align}
\hat{p}[A,\cF] = \sum_{\alpha,\beta} A_{\alpha,\beta} \Scp{\xi}{\pi^{op}(z_\alpha) \pi(y_\beta)\xi}\,,
\end{align}
where \(\pi\) is a representation of the universal enveloping algebra \(\cA\) of \(\Sigma_\infty \), and \(\pi^{op}\) is a representation of the opposite algebra \(\bar{\cA}\). Let \(\cN\) be the von Neumann algebra generated by \(\pi(\cA)\). Since we assume that Connes' embedding conjecture holds, all von Neumann algebras satisfy Kirchberg's QWEP property~\cite{Kirchberg:1993fc} which implies that \(\cN = \cB/\cJ\), where the \(C^*\)-algebra \(\cB\) has the WEP property, and \(\cJ\) is a two-sided ideal in \(\cB\). Since \(y_\beta\) are assumed to be hermitian elements, the Cayley transform \(U_\beta\) of \(\pi(y_\beta)\) is a unitary operator. Let \(\hat{\pi}: C^*[\mathbb{F}_M] \to \cN\) be the *-homomorphism defined by \(s_\alpha \mapsto U_\alpha\), where \(s_\alpha\) are the generators of the free group of \(M\) elements (\(C^*[\mathbb{F}_M]\) is the corresponding universal free group algebra). We apply the same procedure to get another *-homomorphism \(\hat{\pi^{op}}: C^*[\mathbb{F}_M]^{op} \to \pi^{op}(\bar{\cA})\). Now, \(C^*[\mathbb{F}_M]\) as a free group algebra satisfies the Lifting property~\cite{Ozawa_survey1}, and thus the mapping
\begin{align}
\hat{\pi^{op}} \otimes \hat{\pi} : C^*[\mathbb{F}_M]^{op} \otimes C^*[\mathbb{F}_M] \to \pi^{op}(\bar{\cA})\pi(\cA)
\end{align}
is continuous with respect to the minimal tensor product, see~\cite[Proposition 1.3 (iv)]{Kirchberg:1993fc}. Correspondingly, the state \(\omega\) defined by the vector \(\xi\) extends to a state \(\hat{\omega}\) on the minimal tensor product. As in the proof of~\cite[Theorem 28]{Ozawa:2013hn}, we can assume that the induced representation of \( C^*[\mathbb{F}_M]^{op}\) indeed reduces to the opposite representation of \(\hat{\pi}(C^*[\mathbb{F}_M])\) on \(\cH\). Now we know that \( C^*[\mathbb{F}_M]^{op} \otimes C^*[\mathbb{F}_M]\) acts on \(S_2(H) = \cH \otimes \bar{\cH}\) as \((\bar{s} \otimes s)(x) = sx\bar{s} \). Since the state \(\hat{\omega}\) can be approximated by a normal state~\cite{Kadison}, by inverting the Cayley transform we find that for any \(\eps >0\) there exists Hilbert-Schmidt operators \(x_i \in S_2(H)\) an hermitian elements \(\hat{z}_\alpha\), \(\hat{y}_\beta\) such that
\begin{align}
\left|\Scp{\xi}{\pi^{op}(z_\alpha) \pi(y_\beta) \xi} - \sum_i \lambda_i \tr\left[ x^*_i \hat{y}_\beta x_i \hat{z}_\alpha\right]\right| \leq \eps\,.
\end{align}
But since the state \(\omega\) originates from an maximization, we can assume that only one term (say given by \(x \in S_2(H)\) in the above sum is non-zero.It follows from \(\Scp{\xi}{\xi} = 1\) that \(\tr[x^* x] = 1\) and hence we can by an approximation argument assume that \(x\) is of finite rank, wit support projection \(p\). Projecting the hermitian elements \(\hat{z}_\alpha\), \(\hat{y}_\beta\) as well, the form
\begin{align}
z_\alpha \times y_\beta \to \tr\left[ x^* p \hat{y}_\beta p x p \hat{z}_\alpha p\right]
\end{align}
is seen to satisfy all the required constraints. In order to bring it into the form~\eqref{eq:introsigma}, we let \(\sigma = |xp|^2\) and find with \(x = u |x|\) the polar decomposition of \(x\) that \(\tr[\sigma] = \tr[|x|u^* u |x|] = \tr[x^* x p ] =1\) as well as
\begin{align}
\tr\left[ x^* p \hat{y}_\beta p x p \hat{z}_\alpha p\right] = \tr\left[ \sigma^{1/2} \hat{y}_\beta \sigma^{1/2} u\hat{z}_\alpha u^* \right] \,.
\end{align}

\subsection{Completely positive semidefinite cone}\label{app:connescsp}

Theorem~\ref{prop:GNSforspforms} also applies to this case, but we get also from the hierarchy that the constructed state fulfills in addition the cyclicity constraint. More precisely, let \(s\) be the state on \(\bar{\cA} \otimes \cA\) constructed in the proof of theorem~\ref{prop:GNSforspforms}. Note that in this setting, we do not distinguish two kinds of variables and hence \(\cA\) is the free C*-algebra generated by \(N\) positive elements \(z_\alpha\), for \(\alpha \in \{1,\ldots,N\}\).  The cyclicity constraints, which are added to each level also hold for the state \(s\), implying that we have 
\begin{align}
s(\overline{z_\alpha \circ u} \otimes v \circ z_\beta) = s(\overline{u \circ z_\beta} \otimes z\alpha \circ v) \,,
\end{align}
where \(u,v\) are arbitrary words in the variables \(z_\alpha\). Applying this identity recursively to the choice \(z_\beta = \idty\), we find for \(u = z_{\alpha_1} \,z_{\alpha_2}\,\cdots\,z_{\alpha_n}\)
\begin{align}
s(\overline{u} \otimes v ) = s(\overline{z_{\alpha_1} \,z_{\alpha_2}\,\cdots\,z_{\alpha_n}} \otimes v) = s(\overline{z_{\alpha_2}\cdots z_{\alpha_n} }\otimes z_{\alpha_1} \circ v = \ldots = s(\bar{\idty} \otimes u^* \circ v) \,.
\end{align}
Moreover, by the same trick we find
\begin{align}
s(\bar{\idty} \otimes u \circ z_\alpha) = s(\bar{z_\alpha} \otimes u) = s(\bar{\idty} \otimes z_\alpha \circ u)\,,
\end{align}
and hence \(s(\bar{\idty} \otimes u\circ v) = s(\bar{\idty} \otimes v \circ u)\). These equalities can be linear extended to hold for all finite polynomials \(u,v \in \cA\) in the variables \(z_\alpha\), which is a dense subset. They are hence true for all \(u,v \in \cA\). Since \(s\) is a state on \(\bar{\cA} \otimes \cA\), \(s(\bar{\idty} \otimes \cA)\) is a state \(\tau\) on \(\cA\), and the constraints just derive imply that it is a tracial state, \(\tau(ab) = \tau(ba)\) for \(a,b \in \cA\). This is also the state which is constructed by the NPA hierarchy, if we would follow the proof steps mentioned in the main text. 

It follows from these considerations that the limiting point \(p^{*}\) of our SDP~\eqref{eq:sdp_cspplus} can be written as
\begin{align}
p^* = \sum_{\alpha,\beta} A_{\alpha,\beta} \tau(z_\alpha\, z_\beta) \,.
\end{align}
Let \(\pi_\tau\) be the GNS representation of the state \(\tau\), and let \(\pi_\tau(\cA)''\) be the finite von Neumann algebra generated by it. If Connes' embedding conjecture holds, then \(\pi_\tau(\cA)''\) embeds into an ultrapower of the hyperfinite factor, preserving the tracial character of the state. Let \(\theta\) be this embedding. Then we have
\begin{align}
p^* = \sum_{\alpha,\beta} A_{\alpha,\beta} \tau\circ \theta^{-1}( \theta(z_\alpha)\, \theta(z_\beta))\,,
\end{align}
and Burgdorf, Laurent and Piovesan~\cite{BLP15} have shown that matrices of the form \( \tau\circ \theta^{-1}( \theta(z_\alpha)\, \theta(z_\beta))\) belong to the closure of the cone \(\mathcal{CS}_+\).


\section{Generalizations concerning constraint sets and objective functions}\label{sec:generalization}

In the main text we only considered linear inequality constraints on the non-commutative variables (expressed by the set $\cF$).\footnote{We think of the commutativity assumption not as of a constraint, but rather as part of the definition of a quantum bilinear program} However, more general constraint sets can also be studied with our approach. 

In particular, equality constraints can be already included into the free algebra. For example, let $q$ be an irreducible polynomial with variables in $\Sigma_\infty$, such as $q(z) = z^2-z$. The requirement that $q(z_i)=0$, $q(y_j)=0$, $i =1,\ldots,N$, $j=1,\ldots,M$, then corresponds to allowing only projection valued operators. If we denote by $\langle q \rangle$ the ideal in $\Sigma_\infty$ generated by $q$, then we can form the quotient *-algebra $\Sigma_\infty / \langle q \rangle$ which intuitively can be understood as starting with the free *-algebra $\Sigma_\infty$ and then imposing the constraint $q$. We can adopt our procedure for deriving the programs~\eqref{eq:sdp_odd} to this new algebra, by defining the bilinear form~\eqref{eq:bilinear_form} on the new algebra $\Sigma_\infty / \langle q \rangle$ and then following the same procedure as before. However, since the simple monomials are not longer a basis for this quotient algebra, the derivation of levels now relies on first obtaining a monomial basis for $\Sigma_\infty / \langle q \rangle$. This can be achieved if a finite Gr\"obner basis exists and is efficiently computable, as already explained in~\cite[Section 3.5]{PV09}. Alternatively, the equality constraints can also be achieved by requiring that certain matrix elements of $\Omega$ are identified with each other. For example, for the constraint above we would have
\begin{align}\label{eq:}
\Omega_{u,v\circ(i)\circ(i)\circ v} = \Omega_{u,v\circ(i)\circ v}\,,
\end{align}
for words $u,v \in \Sigma_\infty$ and $i=1,\ldots,N+M$. 

Apart from adding polynomial equality constraints, also generalizations concerning the objective functions are possible. Up to now, we only considered the case of bilinear terms. However, terms which are linear in just one variable or constant can be added if we allow for the objective matrix $A$ to have also support on words involving the empty word $\emptyset$. That is, objective functions of the form
\begin{align}
  \label{eq:}
  \sum_{\alpha,\beta} A_{\alpha,\beta} \Omega^n_{\alpha,\beta} + \sum_\alpha a_\alpha \Omega^n_{\alpha,\emptyset} + \sum_\beta b_\beta \Omega^n_{\emptyset,\beta} + c \Omega^n_{\emptyset,\emptyset}
\end{align}
fit into our framework. They correspond to optimizing a functional not only depending on the (quantum) correlations, but also on the marginal distributions.


\section{Entanglement-assisted noisy channel coding}\label{app:coding}

Here we show that for the channel $Z$ defined in Section~\ref{sec:channel}, we have
\begin{align}
\channel^{*}(Z, 1) \geq \frac{1}{2} + \frac{1}{\sqrt{6}} \approx 0.908\,.
\end{align}
For that, we give a quantum protocol using a four dimensional maximally entangled state
\begin{align}
\ket{\psi}:=\frac{1}{2} \sum_{i \in [4]} \ket{i}\otimes\ket{i}\,.
\end{align}
For sending the bit $0$, the sender performs a measurement in the computational basis $E(x|0) = \proj{x}$ and for sending the bit $1$, the sender performs a measurement in the rotated basis $E(x|1) = U \proj{x} U^{\dagger}$ with
\begin{align}
U = \frac{1}{\sqrt{3}} 
\left(
\begin{smallmatrix}
0 & -1 & -1 & 1\\
1 & 0 & 1 & 1 \\
-1 & 1 & 0 & 1 \\
-1 & -1 & 1 & 0
\end{smallmatrix}
\right)\,.
\end{align}
The possible outputs of the channel can be labeled by subsets of the inputs of size $2$. We can write the success probability as
\begin{align}
&\frac{1}{6}  \sum_{x \in [4], x' \neq x} \bra{\psi} \proj{x} \otimes D(0|\{x,x'\}) \ket{\psi} +\frac{1}{6}  \sum_{x \in [4], x' \neq x} \bra{\psi} U\proj{x}U^{\dagger} \otimes (\id - D(0|\{x,x'\}) ) \ket{\psi} \\
&= \frac{1}{2} + \frac{1}{6} \sum_{x \in [4], x' \neq x} \bra{\psi} (\proj{x} - U \proj{x} U^{\dagger}) \otimes D(0|\{x,x'\}) \ket{\psi} \\
&= \frac{1}{2} + \frac{1}{6} \cdot 2 \sum_{\{x,x'\} \in \binom{4}{2}} \bra{\psi} \left(\frac{\proj{x} + \proj{x'}}{2} - U \frac{\proj{x} + \proj{x'}}{2} U^{\dagger} \right) \otimes D(0|\{x,x'\})  \ket{\psi}\,.
\end{align}
By choosing $D(0|\{x,x'\})$ to be an optimal measurement to distinguish between the states
\begin{align}
\frac{1}{2}\Big(\proj{x} + \proj{x'}\Big)\quad\mathrm{and}\quad\frac{1}{2} U\Big(\proj{x} + \proj{x'}\Big)U^{\dagger}\,,
\end{align}
we get a success probability of
\begin{align}
\frac{1}{2} + \frac{1}{6}\cdot\frac{1}{4}\sum_{\{x,x'\} \in \binom{4}{2}} 
\left\|\frac{\proj{x} + \proj{x'}}{2} - U \frac{\proj{x} + \proj{x'}}{2} U^{\dagger} \right\|_1 \ = \frac{1}{2} + \frac{1}{\sqrt{6}}\,.
\end{align} 



\section{Quantum-Proof Randomness Extractors}\label{app:extractor}

Here we give the missing proofs for the claims in Section~\ref{sec:randomness}. The full first level of our SDP hierarchy~\eqref{eq:generic-first-level} for quantum-proof randomness extractors is as follows:
\begin{align}\label{eq:sdp_extractor}
\begin{aligned}
\mathrm{Err}^{\sdp_1}(\mathsf{Ext},k)=\;&\underset{\Omega^1}{\text{maximize}}
& & \frac{1}{2^d} \sum_{i,(s,j)} \left[ \delta_{f_s(i)=j} - \frac{1}{2^{m}}\right] \Omega^1_{(i),(s,j)} \\
& \text{subject to}
& & \Omega^1  \in \mathrm{Pos}(1+2^n + 2^{d+m})\\
& & &\Omega^1_{w,w'}\geq0\quad\forall w,w'\in \Sigma_1\\
& & &\Omega^1_{\emptyset,\emptyset} = 1,\;\Omega^1_{\emptyset,w}=\sum_{i}\Omega^1_{(i),w}\quad\forall w\in \Sigma_1\\
& & & 2^{-k} \Omega^1_{\emptyset,w}\geq\Omega^1_{(i),w}\quad \forall i\in \left[2^n\right],\;\forall w\in\Sigma_1\\
& & &\Omega^1_{\emptyset,w}\geq\Omega^1_{(s,j),w}\quad \forall (s,j)\in \left[2^{d+m}\right],\;\forall w\in\Sigma_1\\
& & &2^{-2k}+\Omega^1_{(i),(i')}\geq2^{-k}\Omega^1_{(i),\emptyset}+2^{-k}\Omega^1_{\emptyset,(i')}\quad\forall i,i'\in \left[2^n\right]\\
& & &1+\Omega^1_{(s,j),(s',j')}\geq\Omega^1_{\emptyset,(s',j')}+\Omega^1_{(s,j),\emptyset}\quad\forall (s,j),(s',j')\in\left[2^{d+m}\right]\\
& & &2^{-k}+\Omega^1_{(i),(s,j)}\geq2^{-k}\Omega^1_{\emptyset,(s,j)}+\Omega^1_{(i),\emptyset}\quad\forall i\in\left[2^n\right],\;\forall (s,j)\in\left[2^{d+m}\right]\,.
\end{aligned}
\end{align}
The upper bound~\eqref{eq:sdp1_extractor} is then immediate by ignoring some constraints.

\begin{proof}[Proof of Theorem~\ref{thm:extractor}]
The ideas for the proof are from~\cite[Theorem 5]{berta15}. We first prove~\eqref{eq:extractor_thm1}. For that we relax the positivity constraint in~\eqref{eq:sdp1_extractor} from
\begin{align}
\Omega^1_{(s,j),(s,j)}\geq\quad\mathrm{to}\quad\Omega^1_{(s,j),(s,j)}\geq-1\quad\forall (s,j)\in\left[2^{m+d}\right]\,,
\end{align}
and ignore some of the other constraints in~\eqref{eq:sdp1_extractor} leading to,
\begin{align}
\begin{aligned}
\mathrm{Err}^{\overline{\sdp}_1}(\mathsf{Ext},k)\leq\;&\underset{\Omega^1}{\text{maximize}}
& & \frac{1}{2^d} \sum_{i,(s,j)} \left[ \delta_{f_s(i)=j} - \frac{1}{2^{m}}\right] \Omega^1_{(i),(s,j)} \\
& \text{subject to}
& & \Omega^1  \in \mathrm{Pos}(1+2^n + 2^{d+m})\\
& & &0\leq\Omega^1_{(i),(i')}\leq2^{-k} \sum_i\Omega^1_{(i),(i')}\quad\forall i,i'\in\left[2^n\right]\\
& & &\sum_i\Omega^1_{(i),(i')}\leq2^{-k}\quad\forall i'\in\left[2^n\right]\\
& & &\sum_{i,i'}\Omega^1_{(i),(i')}=1\\
& & &-1\leq\Omega^1_{(s,j),(s,j)}\leq1\quad\forall (s,j)\in\left[2^{m+d}\right]\,.
\end{aligned}
\end{align}
Moreover, we write $\Omega^1  \in \mathrm{Pos}(1+2^n + 2^{d+m})$ as a Gram matrix:
\begin{align}
\Omega^1_{u,u'}=:\vec{a}_u\cdot\vec{a}_{u'},\;\Omega^1_{u,v}=:\vec{a}_u\cdot\vec{b}_v,\;\Omega^1_{v,v'}=:\vec{b}_v\cdot\vec{b}_{v'}\quad\forall u,u'\in\left[2^n\right]\cup\{\emptyset\},\;\forall v,v'\in\left[2^{m+d}\right]\cup\{\emptyset\}\,.
\end{align}
An application of the Cauchy-Schwarz inequality then easily reveals that the optimal choice for $\vec{b}_{(s,j)}$ is
\begin{align}
\vec{b}_{(s,j)}=\frac{\sum_i\left[ \delta_{f_s(i)=j}-\frac{1}{2^m} \right]\vec{a}_i}{\left\| \sum_i\left[\delta_{f_s(i)=j}-\frac{1}{2^m}\right]\right\|_2}\,.
\end{align}
Thus, the upper bound program becomes
\begin{align}\label{eq:mbound_sdp}
\begin{aligned}
\mathrm{Err}^{\overline{\sdp}_1}(\mathsf{Ext},k)\leq\;&\underset{\vec{a}_i}{\text{maximize}}
& & \frac{1}{2^d} \sum_{(s,j)}\left\| \sum_i \left[ \delta_{f_s(i)=j} - \frac{1}{2^{m}}\right]\vec{a}_i\right\|_2\\
& \text{subject to}
& &0\leq\vec{a}_i\cdot\vec{a}_{i'}\leq2^{-k}\sum_i\vec{a}_i\cdot\vec{a}_{i'}\quad\forall i,i'\in \left[2^n\right]\\
& & &\sum_i\vec{a}_i\cdot\vec{a}_{i'}\leq2^{-k}\quad\forall i'\in \left[2^n\right]\\
& & &\sum_{i,i'}\vec{a}_i\cdot\vec{a}_{i'}=1\,.
\end{aligned}
\end{align}
Again using the Cauchy-Schwarz inequality we can write
\begin{align}
\frac{1}{2^d} \sum_{(s,j)} \left\| \sum_{i} \left[\delta_{f_s(i)=j} - \frac{1}{2^m} \right] \vec{a}_i \right\|_2\leq\sqrt{ \frac{1}{2^d} \sum_{(s,j)} \left\| \sum_{i} \left[\delta_{f_s(i)=j} - \frac{1}{2^m} \right] \vec{a}_i \right\|_{2}^2}  \sqrt{2^m} \ .
\end{align}
Letting (again),
\begin{align}
\Omega^1_{(i),(i')}=\vec{a}_i\cdot\vec{a}_{i'}\quad\mathrm{and}\quad\bar{\Omega}^1_{(i)}\eqdef\sum_{i'}\Omega^1_{(i),(i')}\,,
\end{align}
we look at the expression
\begin{align}
\frac{1}{2^d} \sum_{(s,j)} \left\| \sum_{i} \left[\delta_{f_s(i)=y} - \frac{1}{2^m} \right]\vec{a}_i \right\|_{2}^2&=\frac{1}{2^d}\sum_{(s,j)}\sum_{i,i'}\left[\delta_{f_s(i)=j} - \frac{1}{2^m} \right]\cdot\left[\delta_{f_s(i')=j} - \frac{1}{2^m} \right]\Omega^1_{(i),(i')}\\
&\leq\frac{1}{2^d}\sum_{(s,j)}\sum_i\bar{\Omega}^1_{(i)}\left|\sum_{i'}\left[\delta_{f_s(i')=j} - \frac{1}{2^m} \right]\frac{\Omega^1_{(i),(i')}}{\bar{\Omega}^1_{(i)}}\right|\\
&\leq\frac{1}{2^d}\sum_{(s,j)}\left|\sum_{i}\delta_{f_s(i)=j}2^{-k} - \frac{1}{2^m}\right|\,,
\end{align}
where we made use of the constraints in~\eqref{eq:mbound_sdp} for the last inequality. Going back to the error $\mathrm{Err}(\mathsf{Ext},k)$ as in~\eqref{eq:extractor_1norm} we conclude the claim.

We now prove~\eqref{eq:extractor_thm2}. We upper bound $\mathrm{Err}^{\overline{\sdp}_1}(\mathsf{Ext},k)$ by forgetting several constraints and then apply Grothendieck's inequality (see Lemma~\ref{lem:grothendieck} below):
\begin{align}
\mathrm{Err}^{\overline{\sdp}_1}(\mathsf{Ext},k) &\leq \max \left\{ \frac{1}{2^d} \sum_{i,(s,j)} \left(\delta_{f_s(i) = j} - \frac{1}{2^m} \right) \vec{a}_{i} \cdot \vec{b}_{(s,j)} : \left\| \vec{a}_{i} \right\|_2 \leq 2^{-k}, \left\| \vec{b}_{(s,j)} \right\|_{2} \leq 1 \right\} \\
&\leq K_G \max \left\{ \frac{1}{2^d} \sum_{i,(s,j)} \left(\delta_{f_s(i) = j} - \frac{1}{2^m} \right) a_i b_{(s,j)} : |a_i| \leq 2^{-k}, | b_{(s,j)} | \leq 1 \right\} \\
&\leq K_G \max \left\{ \frac{1}{2^d} \sum_{(s,j)} \left| \sum_{i} \left(\delta_{f_s(i) = j} - \frac{1}{2^m} \right) a_i \right| : |a_i| \leq 2^{-k} \right\}\,.
\end{align}
We partition the set of $i \in\left[2^n\right]$ into $\{ i : a_i \geq 0\}$ and $\{ i : a_i < 0\}$, and write 
\begin{align}
\left| \sum_{i} \left(\delta_{f_s(i) = j} - \frac{1}{2^m} \right) a_i \right|
&\leq \left| \sum_{i : a_i \geq 0} \left(\delta_{f_s(i) = j} - \frac{1}{2^m} \right) a_i \right|+\left|\sum_{i : a_i < 0} \left(\delta_{f_s(i) = j} - \frac{1}{2^m} \right) (-a_i) \right|\,.
\end{align}
Let us write
\begin{align}
\alpha_+ \eqdef \sum_{i : a_i \geq 0} a_{i}\,.
\end{align}
Now if $\alpha_{+} \geq 1$, then we define
\begin{align}
p_{+}(i)\eqdef\frac{\max\{a_i,0\}}{\alpha_+}\,.
\end{align}
Observing that $\alpha_+ \leq 2^{n-k}$, we have
\begin{align}
\frac{1}{2^d} \sum_{(s,j)} \left| \sum_{i : a_i \geq 0} \left(\delta_{f_s(i) = j} - \frac{1}{2^m} \right) a_i \right|&= \alpha_+ \cdot \frac{1}{2^d} \sum_{(s,j)} \left| \sum_{i} \left(\delta_{f_s(i) = j} - \frac{1}{2^m} \right) p_{+}(x) \right| \\
&\leq 2\alpha_{+} \mathrm{Err}(\mathsf{Ext},k+\log(\alpha_+)) \\
&\leq 2\cdot2^{n-k} \mathrm{Err}(\mathsf{Ext},k)\,,
\end{align}
with the error $\mathrm{Err}(\mathsf{Ext},k)$ as in~\eqref{eq:extractor_1norm}. Otherwise, if $\alpha_{+} < 1$, then we define
\begin{align}
p_{+}(i)\eqdef\max\{a_i, 0\} + (1-\alpha_{+})2^{-n}\,.
\end{align}
We have
\begin{align}
&\frac{1}{2^d} \sum_{(s,j)} \left| \sum_{i : a_i \geq 0} \left(\delta_{f_s(i) = j} - \frac{1}{2^m} \right) a_i \right|\notag\\
&= \frac{1}{2^d} \sum_{(s,j)} \left| \sum_{i} \left(\delta_{f_s(i) = j} - \frac{1}{2^m} \right) (p_{+}(i) - (1-\alpha_+)2^{-n}) \right| \\
&\leq \frac{1}{2^d} \sum_{(s,j)} \left| \sum_{i} \left(\delta_{f_s(i) = j} - \frac{1}{2^m} \right) p_{+}(i) \right| + (1-\alpha_+) \frac{1}{2^d} \sum_{(s,j)} \left| \sum_{i} \left(\delta_{f_s(i) = j} - \frac{1}{2^m} \right) 2^{-n} \right|  \\
&\leq 2\mathrm{Err}(\mathsf{Ext},k-1) + 2(1-\alpha_+) \mathrm{Err}(\mathsf{Ext},n)\,.
\end{align}
With a similar argument for the set $\{ i : a_i < 0\}$, we reach the bound
\begin{align}
&\frac{1}{2^d} \sum_{(s,j)} \left| \sum_{i} \left(\delta_{f_s(i) = j} - \frac{1}{2^m} \right) a_i \right| \\
&\leq 2\max\Big\{ 2 \cdot 2^{n-k} \mathrm{Err}(\mathsf{Ext},k), \mathrm{Err}(\mathsf{Ext},k-1) + \mathrm{Err}(\mathsf{Ext},n) + 2^{n-k}\mathrm{Err}(\mathsf{Ext},k),\notag\\
&\qquad\qquad\quad2\mathrm{Err}(\mathsf{Ext},k-1) + (1-\alpha_+ - \alpha_-) \mathrm{Err}(\mathsf{Ext},n) \Big\} \\
&\leq6\cdot2^{n-k} \mathrm{Err}(\mathsf{Ext},k-1)\,.
\end{align}
From this we conclude the claim
\begin{align}
\mathrm{Err}^{\overline{\sdp}_1}(\mathsf{Ext},k)\leq6\cdot2^{n-k}\mathrm{Err}(\mathsf{Ext},k-1)\,.
\end{align}
\end{proof}

\begin{lemma}[Grothendieck's inequality]\label{lem:grothendieck}
For any real matrix $\{A_{ij}\}$, we have
\begin{align}
\max \left\{ \sum_{i,j} A_{ij} \vec{a}_{i} \cdot \vec{b}_j : \left\| \vec{a}_{i} \right\|_2 \leq 1, \left\| \vec{b}_j \right\|_2 \leq 1 \right\}\leq K_G \cdot \max \left\{ \sum_{i,j} A_{ij} a_i b_j : a_i, b_j \in \bR, |a_i| \leq 1, |b_j| \leq 1 \right\}\,.
\end{align}
\end{lemma}


\section*{Acknowledgments}
We thank Hamza Fawzi and Thomas Vidick for helpful discussions. We would also like to thank Alhussein Fawzi and Hamza Fawzi for their help with writing and running MATLAB code. Part of this work was done while VBS was visiting the \'Ecole Normale Sup\'erieure de Lyon and part of it while visiting the Institute for Quantum Information and Matter at Caltech, we thank John Preskill and Thomas Vidick for their hospitality. MB acknowledges funding provided by the Institute for Quantum Information and Matter, an NSF Physics Frontiers Center (NFS Grant PHY-1125565) with support of the Gordon and Betty Moore Foundation (GBMF-12500028). Additional funding support was provided by the ARO grant for Research on Quantum Algorithms at the IQIM (W911NF-12-1-0521). OF acknowledges support by the LABEX MILYON (ANR-10-LABX-0070) of Universit\'e de Lyon, within the program ``Investissements d'Avenir" (ANR-11-IDEX-0007) operated by the French National Research Agency (ANR).
VBS acknowledges partial support by the EU project Randomness and Quantum Entanglement (RAQUEL) and the NCCR QSIT. 

\end{document}